\documentclass{sn-jnl}
\usepackage[activeacute,english]{babel} 
\usepackage[utf8]{inputenc}
\usepackage{verbatim} 
\usepackage{amsfonts,amssymb,mathtools} 

\usepackage[labelfont=bf,labelsep=space]{caption}
\usepackage[algo2e,vlined,ruled]{algorithm2e}
\SetKwInOut{postcondition}{Postcondition}
\SetKwInOut{precondition}{Precondition}

\makeatletter
\let\original@algocf@latexcaption\algocf@latexcaption
\long\def\algocf@latexcaption#1[#2]{%
  \@ifundefined{NR@gettitle}{%
    \def\@currentlabelname{#2}%
  }{%
    \NR@gettitle{#2}%
  }%
  \original@algocf@latexcaption{#1}[{#2}]%
}
\makeatother

\numberwithin{equation}{section} 
\usepackage[labelfont=bf,labelsep=space]{caption} 

\usepackage[OT2,OT1]{fontenc}
\DeclareRobustCommand\cyr{%
  \renewcommand\rmdefault{wncyr}%
  \renewcommand\sfdefault{wncyss}%
  \renewcommand\encodingdefault{OT2}%
  \normalfont
  \selectfont}
\DeclareTextFontCommand{\textcyr}{\cyr}

\def\mcC{\mathcal{C}}
\def\mcF{\mathcal{F}}

\def\mcH{\mathcal{H}}

\def\mcM{\mathcal{M}}
\def\mcN{\mathcal{N}}
\def\mcP{\mathcal{P}}

\def\mcS{\mathcal{S}}

\def\bbE{\mathbb{E}}
\def\bbF{\mathbb{F}}

\def\bbR{\mathbb{R}}
\def\bbZ{\mathbb{Z}}
\def\bbN{\mathbb{N}}

\def\bbC{\mathbb{C}}
\def\bbP{\mathbb{P}}
\def\bbS{\mathbb{S}}
\def\bfm{\mathbf{m}}
\def\bfu{\mathbf{u}}
\def\fo{\mathrm{R}}

\def\fkc{\mathfrak{c}}
\def\fkf{\mathfrak{f}}
\def\fkg{\mathfrak{g}}

\def\fkq{\mathfrak{q}}
\def\fkx{\mathfrak{x}}

\def\fky{\mathfrak{y}}

\def\yuproj{\textrm{\cyr Yu}}

\DeclareMathOperator{\vol}{vol}

\def\Hd{\mcH_{n,d}}
\def\Pd{\mcP_{n,d}}

\def\hm{^{\mathsf{h}}}

\def\kappaff{\kappa_{\sf aff}}

\DeclareMathOperator{\Oh}{\mathcal{O}}
\def\Tg{\mathrm{T}}
\def\rmD{\mathrm{D}}
\def\diff{\rmD}
\def\Diff{\rmD}
\newcommand{\errsymb}[1]{\theta_{#1}}
\newcommand{\eps}{\varepsilon}

\renewcommand{\hat}{\widehat}

\newcommand{\fl}{\mathtt{fl}}

\def\transp{^{\mathrm{T}}}
\def\ll{{[\kern-1.6pt [}}
\def\rr{{]\kern-1.4pt ]}}
\def\bll{{\biggl[\kern-3pt \biggl[}}
\def\brr{{\biggr]\kern-3pt \biggr]}}

\DeclarePairedDelimiter\abs{\lvert}{\rvert}%
\DeclarePairedDelimiter\norm{\lVert}{\rVert}%

\newcommand{\eproof}{\hfill\qed\smallskip}

\makeatletter
\let\oldabs\abs
\def\abs{\@ifstar{\oldabs}{\oldabs*}}
\let\oldnorm\norm
\def\norm{\@ifstar{\oldnorm}{\oldnorm*}}
\makeatother

\definecolor{red}{rgb}{.7,0,0}
\definecolor{blue}{rgb}{0,0,1}

\newcommand{\rojito}[1]{#1}

\jyear{2021}%

\theoremstyle{thmstyleone}%
\newtheorem{theorem}{Theorem}[section]
\newtheorem{proposition}[theorem]{Proposition}%
\newtheorem{lemma}[theorem]{Lemma}
\newtheorem{corollary}[theorem]{Corollary}

\theoremstyle{thmstyletwo}%
\newtheorem{remark}[theorem]{Remark}%

\theoremstyle{thmstylethree}%
\newtheorem{definition}[theorem]{Definition}%

\begin{document}
\title[On the Complexity of the Plantinga-Vegter Algorithm]{On the Complexity of the Plantinga-Vegter Algorithm }

\author[1]{\fnm{Felipe} \sur{Cucker}}\email{macucker@cityu.edu.hk}

\author[2]{\fnm{Alperen A.} \sur{Erg\"{u}r}}\email{alperen.ergur@utsa.edu}

\author*[3]{\fnm{Josu\'{e}} \sur{Tonelli-Cueto}}\email{josue.tonelli.cueto@bizkaia.eu}\equalcont{All the authors contributed equally to this work.}

\affil[1]{\orgdiv{Department of Mathematics}, \orgname{City University of Hong Kong}, \orgaddress{\country{Hong Kong}}}

\affil[2]{\orgdiv{Mathematics Department}, \orgname{University of Texas at San Antonio}, \orgaddress{\street{One UTSA Circle}, \city{One UTSA Circle}, \postcode{78249}, \state{Texas}, \country{USA}}}

\affil*[3]{\orgdiv{OURAGAN team}, \orgname{Inria Paris \& IMJ-PRG,\\ Sorbonne Universit\'{e}}, \orgaddress{\city{Paris}, \country{France}}}

\abstract{
\rojito{We introduce tools from numerical analysis and high dimensional probability for precision control and complexity analysis of subdivision-based algorithms in computational geometry. We combine these tools with the continuous amortization framework from exact computation. We use these tools on a well-known example from the subdivision family: the adaptive subdivision algorithm due to Plantinga and Vegter. The only existing complexity estimate on this rather fast algorithm was an exponential worst-case upper bound for its interval arithmetic version. We go beyond the worst-case by considering both average and smoothed analysis, and prove polynomial time complexity estimates for both interval arithmetic and finite-precision versions of the Plantinga-Vegter algorithm.}}

\maketitle



\section{Introduction}

Subdivision based algorithms are ubiquitous in computational geometry. These algorithms have the advantage of simplicity, and often have good practical performance. The two main challenges related to subdivision based algorithms are the control of precision (or a termination criterion), and complexity analysis. As late as summer 2019, complexity analysis aspect of subdivision based geometric algorithms was considered to be ``largely open'' \cite{yap2019towards}. In this paper, we contribute to both of the main challenges by introducing a hybrid toolbox that combines condition numbers, high dimensional probability theory, and continuous amortization framework introduced by Burr, Krahmer, and Yap ~\cite{burr2009}. To keep our writing focused, and the length of the article finite, we only showcase the toolbox on a well-known member of this large family; the algorithm of Plantinga and Vegter. 

Plantinga-Vegter (PV) algorithm is an adaptive subdivison algorithm for meshing curves and surfaces ~\cite{plantingavegter2004}. The algorithm admits an implicit equation of a curve or a surface and outputs an isotopic piecewise linear approximation with controlled Hausdorff distance. The initial paper of Plantinga and Vegter contained no complexity analysis and not even a formal setting fixing either the kind of functions implicitly defining the considered curves and surfaces or the arithmetic used. However, concrete implementations in the paper indicated the efficiency of the algorithm. The algorithm is now widely considered to be very efficient.

The first complexity analysis of the PV algorithm was published thirteen years later by Burr, Gao and Tsigaridas~\cite{burr2017} (cf.~\cite{burr2020}). The paper of Burr, Gao, and Tsigaridas focused on the subdivision procedure of the Plantinga-Vegter algorithm and only analyzed the complexity for polynomials with integer coefficients. The paper provides bounds that are exponential both in the degree $d$ of the input  polynomial and in its logarithmic height $\tau$. The discrepancy between the exponential complexity estimate and the practical efficiency of the PV algorithm was marked by the following comment at the end of the paper 
\begin{quote}
  Even though our bounds are optimal, in practice, these are quite
  pessimistic [\dots]
\end{quote}
The authors further observe that, following from their Proposition~5.2 
(see Theorem~\ref{theo:analysis2} below) an instance-based analysis of the algorithm (i.e., one yielding a cost that 
depends on the input at hand) could be derived
from the evaluation of a certain integral. And they conclude their paper by writing  
\begin{quote}
  Since the complexity of the algorithm can be exponential in the inputs [size], the integral must be described in terms of additional geometric and intrinsic parameters. 
\end{quote}

In this paper, we make progress towards these aims by going 
beyond the worst-case analysis and by using condition numbers. 
We believe 
condition numbers are a perfect fit for the latter aim as they provide a geometric and arguably intrinsic parameter.

\rojito{We analyze the complexity of the PV algorithm in two different
versions corresponding, roughly speaking, to its arithmetic complexity and its (arguably more realistic) bit complexity.  
Our analysis deals with the subdivision routine of the PV
algorithm for curves and surfaces as the special cases for $n=2$
and $n=3$, but we aim for estimates that hold for any $n$.  
We perform both average and smoothed analysis for the two 
versions of the PV algorithm, so we provide four different complexity analyses.}

\rojito{The average analysis framework is well-known. The smoothed
analysis framework might, in contrast, require a bit of an
explanation. Suppose we endow the space of $n$-variate degree 
$d$ polynomials with a norm $\norm{~}$, and a probability 
measure $\mu$ (with as few assumptions as possible on $\mu$).
Suppose $g$ is a random polynomial distributed with respect to
$\mu$. Then we consider an arbitrary polynomial $f$, and we fix 
a tolerance parameter $\sigma >0$. We consider 
$q= f + \sigma \norm{f} g$ as random perturbation of $f$ with
tolerance $\sigma$, and conduct average analysis of the PV
algorithm for $q$. This type of estimate could a priori 
depend on the arbitrary polynomial $f$. We aim for a uniform
estimate that provides an upper bound for any $f$, and depends
only on $\sigma$, $n$, and $d$. This uniform upper bound will be
the smoothed analysis of the PV algorithm. It turns out that  
this random perturbation idea was already considered in the 
computational geometry literature in an experimental fashion, 
and there were aims for building a theoretical framework 
(see section 4 of~\cite{Funke}).}

Our main results Theorem~\ref{thm:probMAIN} and 
Theorem~\ref{thm:probMAINFP} provide the four promised estimates
on the complexity of the PV algorithm for any number of variables
$n$. For the special case of the plane curves, the average and
smoothed analysis of the arithmetic complexity of the PV 
algorithm are respectively 
$\Oh(d^7)$ and $\Oh\left(d^7 (1+ \frac{1}{\sigma})^3 \right)$. 
The average and smoothed analysis of the bit complexity are just slightly worse: $\Oh(d^7\log^2d)$ and 
$\Oh(d^7\log^2d\;(1+\frac1\sigma)^3)$, respectively.  
These bounds are in marked contrast with the 
$\Oh(2^{\tau d^4\log d})$ {\em worst-case} complexity bound
in~\cite{burr2017}.

For a clear presentation of  our contribution and \rojito{related} 
complexity considerations we need to make a few remarks:
\smallskip

\noindent{\bf (1)}
The use of floating-point arithmetic generates numerical 
errors which accumulate during the computation. An important remark is that, despite this accumulation of errors, our algorithm 
returns a correct output, a subdivision with the properties we want. It is, in this sense, a {\em certified} algorithm. 
At the heart of this remark is the fact that  a 
sufficiently small perturbation of a correct subdivision is still a correct subdivision for a generic (i.e. non-singular) input. Condition numbers allow us to estimate how large  this perturbation may be. Then, the fact that we can  estimate these condition numbers,  we control the precision of the operations'  round-off, and we know how these operations are sequenced  further allows us to ensure that the subdivision we constructed is close enough to the one we would have done in an error-free context 
\rojito{and both yield polygons with the same isotopy type.}

Needless to say, for input data outside the set satisfying the 
generic property above our reasoning does not hold. The set of 
such inputs, referred to as {\em ill-posed} in numerical analysis, has measure zero. Condition numbers relate to ill-posedness 
in the sense that 
the closer a data is to the set of ill-posed inputs the larger 
becomes its condition number. It is these facts that allows one to establish average and smoothed analysis by means of probabilistic estimates on the condition numbers. This general scheme was proposed in~\cite{Smale97}.  A more detailed discussion of these issues is 
in~\cite[\textsection 9.5]{Condition}. A relatively early case of  
a fully studied variable-precision algorithm is 
in~\cite{CP01}. An account of the use of floating-point arithmetic in computational geometry is given in~\cite{Funke}.

\smallskip

\noindent{\bf (2)}
Most of the probabilistic analyses for cost measures 
or condition numbers use the Gaussian measure. This choice is mainly for technical convenience. For the analysis of condition numbers, this goes back to  Goldstine and von Neumann~\cite{vNGo51} and, more recently, resulted in simple bounds for a large class of condition  numbers~\cite{Demmel88,BuCuLo:06b,BuCuLo:07,lotz2015}.

In the last few years, however, the search for more robust complexity analysis resulted in estimates that hold for a (quite) general family of measures. The family  of {\em subgaussian} measures which includes all compactly supported random variables provides a good testing ground. An analysis of a condition number  for these distributions occupies~\rojito{\cite{EPR18,EPR19}}. It is for this class of distributions (subgaussians with an anti-concentration property) that our results\rojito{, both average
and smoothed,} are proved. 
\smallskip

\noindent{\bf (3)}
The subdivision procedure we analyze can be considered at three levels of generality: the {\em abstract}, in which we only take into account the number of iterations of the subdivision procedure; the {\em interval}, in which we take also into account the number of arithmetic operations; and the {\em effective}, in which we take into account not only the number of arithmetic operations, but also the precision that they need, obtaining a realistic estimation of the bit-cost of the algorithm. This division follows a trend for analysing subdivision algorithms initiated by Xu and Yap~\cite{xuyap2019} (cf.~\cite{yap2019towards}). 

Our condition-based analysis can be applied at each of these three 
levels, hopefully showing the usefulness of the approach.
Whereas this paper focuses on a particular subdivision procedure we believe that the techniques in this paper can be readily applied to other subdivision based algorithms in computational geometry. We note, however,  that the complexity analysis in 
this paper would have been impossible without the~\emph{continuous amortization} technique developed in the exact numerical  
context~\cite{burr2009,burr2016}. In this regard, 
we hope to trigger a fruitful exchange of ideas between the different approaches to continuous computation and improve our (seemingly preliminary) understanding of the complexity of subdivision algorithms in computational geometry.

\section*{Outline}

The rest of the paper is structured as follows: We start with a section that contains notation. We beg readers' pardon for this inconvenient start; this seemed the simplest way for getting things clear. Then in Section~\ref{sec:pv} we discuss the Plantinga-Vegter algorithm and the $n$-dimensio\-nal generalization of its subdivision method in the abstract, the interval arithmetic, and the effective versions. Section~\ref{sec:main} introduces our randomness model and contains main complexity estimates of this  paper. In Section~\ref{sec:geomfram}, we present a geometric framework (read Hilbert space structure) to deal with homogeneous polynomials. In Section~\ref{sec:condition}, we introduce the condition number $\kappaff$ ---both local, i.e., at a point $x$, and global--- along with its main properties. In Section~\ref{sec:complexity}, we present the existing results 
on the complexity of Plantinga-Vegter algorithm  from~\cite{burr2017}, and we relate these results to the local condition number. In Section~\ref{sec:FP}, we carry out the finite-precision analysis  deriving the corresponding bounds for bit-cost. Finally, in Section~\ref{sec:probability}, we derive average and smoothed complexity bounds under (quite) general randomness assumptions. 

\section*{Notation}

\rojito{Throughout the paper, we will assume some familiarity with the basics of differential geometry.
For a smooth map $f:\bbR^m\rightarrow \bbR$,  $\diff_xf:\Tg_x\bbR^m\cong\bbR^m\rightarrow \Tg_x\bbR\cong\bbR$ denotes the tangent map of $f$ at $x\in\bbR^m$. We will write $\partial f:\bbR^m\rightarrow\bbR^m$, 
$x\mapsto \partial f(x)$ when we see it as a smooth function 
of~$x$.
When we want to see $\partial f$ as a vector of formal derivatives, we will write $\partial f(X)$ where $X$ represents formal variables.}
For general smooth maps between smooth manifolds $F:\mcM\rightarrow \mcN$, we
will just write $\Diff_xF:\Tg_x\mcM\rightarrow \rojito{\Tg_{F(x)}}\mcN$ as the
tangent map.

In what follows, $\Pd$ will denote the set of real polynomials in the $n$ variables $X_1,\ldots,X_n$ with degree at most $d$, 
$\Hd$ the set of homogeneous real polynomials in the $n+1$
variables $X_0,X_1,\ldots,X_n$ of degree $d$, and  $\|~\|$
and 
$\langle\,~,~\rangle$ will denote the standard norm
and inner product in $\bbR^m$ as well as the Weyl norm and
inner product in $\Pd^m$ and $\Hd^m$. 
Given a polynomial $f\in \Pd$, $f\hm\in\Hd$ will be its
homogenization and $\partial f$ the polynomial map given 
by its partial derivatives. We will denote by the Cyrillic
character $\yuproj$, 'yu', the central
projection~\eqref{eq:juproj} that maps $\bbR^n$ into
$\bbS^n$. For details see Section~\ref{sec:geomfram}.
Additionally, $V_\bbR(f)$ and $V_\bbC(f)$ will be,
respectively, the real and complex zero sets of $f$.

\rojito{For a set $S\subset\bbR^n$,} we will denote by
\rojito{$\square S$} the set of $n$-boxes of the form
$x+I^n$, where $I$ is an interval, that are contained in
\rojito{$S$} and, for a given box $B\in\square \bbR^n$,
$m(B)$ will be its middle point, $w(B)$ its width, and 
$\vol B=w(B)^n$ its volume. 

Regarding probabilistic conventions, we will denote the
probability of an event by $\bbP$, random variables by
$\fkx,\fky,\ldots$ and random polynomials by
$\fkf,\fkg,\fkq,\ldots$ The expression 
$\bbE_{\fkx\in K}g(\fkx)$ will denote the expectation 
of $g(\fkx)$ when $\fkx$ is sampled uniformly from the 
set $K$ and $\bbE_{\fky}g(\fky)$ the
expectation of $g(\fky)$ with respect to a previously
specified probability distribution of $\fky$. 

Regarding complexity parameters, $n$ will be the
number of variables, $d$ the degree bound, and $N=\binom{n+d}{n}$ the dimension of $\Pd$. Finally, $\ln$ will denote the natural 
logarithm and $\log$ the logarithm in base $2$.

\section{The Plantinga-Vegter (Subdivision) Algorithm}\label{sec:pv}

Given a real smooth hypersurface in $\bbR^n$ described implicitly by a map $f:\bbR^n\rightarrow \bbR$ and a region $[-a,a]^n$, the Plantinga-Vegter Algorithm constructs a piecewise-linear approximation of the intersection of its zero set $V_\bbR(f)$ with $[-a,a]^n$ isotopic to this intersection inside $[-a,a]^n$. The Plantinga-Vegter algorithm (see Figure~\ref{fig:PVilustration} for an illustration\footnote{This figure is taken from~\cite[Figure~5\textsuperscript{\textsection 1}]{tonellicuetothesis}.}) is divided in two phases:
\begin{enumerate}[1)]
    \item Subdivision phase: In this phase, the Plantinga-Vegter algorithm subdivides $[-a,a]^n$ into smaller and smaller boxes until all the boxes satisfy a certain condition (see~\eqref{eq:conditionPV}).
    \item Post-processing phase: In this phase, the Plantinga-Vegter algorithm uses the obtained subdivision to produce a piecewise-linear approximation of the given hypersurface.
\end{enumerate}
We will focus on the subdivision phase of the Plantinga-Vegter algorithm. We do this because the complexity of subdivision-based algorithms is usually dominated by the complexity of the subdivision phase. This follows the guidelines of the first complexity analysis given by Burr, Gao and Tsigaridas~\cite{burr2017} (cf.~\cite{burr2020}).

We note that it would be interesting to incorporate  the complexity of the post-processing phase of the algorithm to our estimates in this paper: either the original one by Plantinga-Vegter~\cite{plantingavegter2004}, for $n\leq 3$, or the generalization to higher dimensions by Galehouse~\cite{galehousethesis}, for arbitrary $n$. We also don't cover existing extensions of the Plantinga-Vegter algorithm to singular curves~\cite{burr2012}.

\begin{figure}
\centering
\begin{tabular}{c c}
\includegraphics[width=0.42\textwidth]{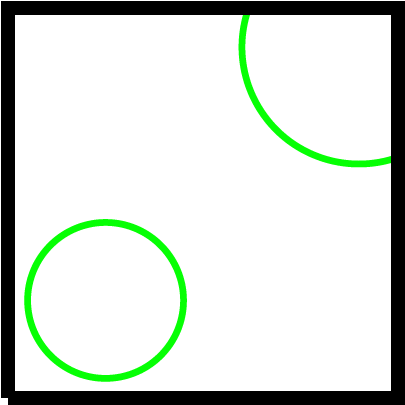}
&
\includegraphics[width=0.42\textwidth]{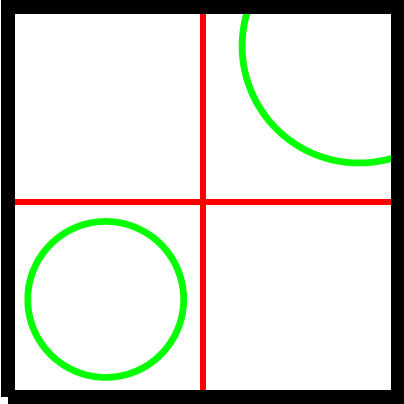}
\\
Step 0 of subdivision phase&Step 1 of subdivision phase\\[6pt]
\includegraphics[width=0.42\textwidth]{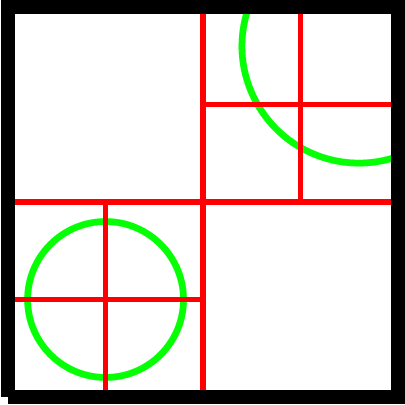}
&
\includegraphics[width=0.42\textwidth]{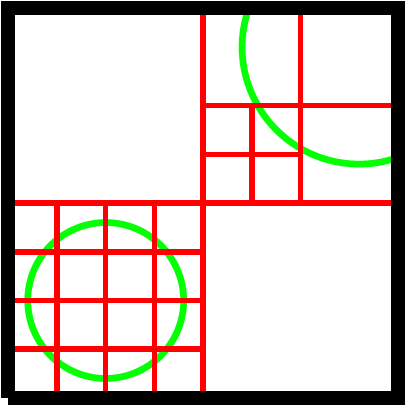}
\\
Step 2 of subdivision phase&Step 4 of subdivision phase\\[6pt]
\multicolumn{2}{c}{
\includegraphics[width=0.42\textwidth]{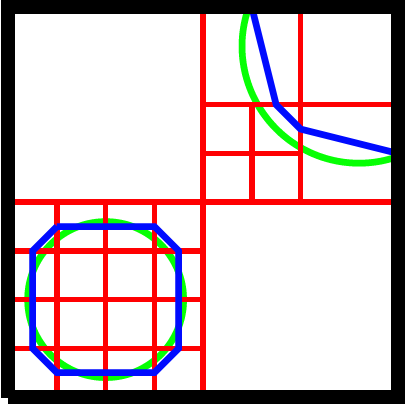}
}\\
\multicolumn{2}{c}{Post-processing phase}\\[6pt]
\multicolumn{2}{c}{Green: $V_\bbR(f)$\hfill Red: Subdivision\hfill Blue: PL approximation of $V_\bbR(f)$}
\end{tabular}
\caption{Plantinga-Vegter applied to $f=X^4 - 6 X^3 + 2 X^2 Y^2 - 6 X^2 Y - 34 X^2 - 6 X Y^2 - 320 X Y + 376 X + Y^4 - 6 Y^3 - 34 Y^2 + 376 Y + 3128$ in $[-10,10]^2$.}
\label{fig:PVilustration}
\end{figure}

From now on, when we say Plantinga-Vegter algorithm we are referring to the Plan\-tin\-ga-Vegter subdivision phase and,
\rojito{following~\cite{burr2017}}, we restrict to the case in which $f:\bbR^n\rightarrow \bbR$ is a polynomial. We now describe this algorithm at three levels: abstract, interval and effective. 

\subsection{Abstract level: Algorithm~\texorpdfstring{\mdseries{\nameref{alg:PVAlgorithm}}}{\nameref{alg:PVAlgorithm}}}

The Plantinga-Vegter algorithm subdivides $[-a,a]^n$ until a certain regularity condition is satisfied in each of the 
boxes $B$ of the subdivision. Let $h,\rojito{\tilde{h}}:\bbR^n\rightarrow (0,\infty)$ be some fixed positive maps, conveniently chosen (see~\eqref{eq:definitonhhprime} 
\rojito{and Remark~\ref{rem:h}} 
below). Then this regularity condition is  
\begin{equation}\label{eq:conditionPV}
    C_f(B)\text{: either }0\notin (hf)(B)\text{ or }0\notin \langle (\rojito{\tilde{h}}\partial f)(B),(\rojito{\tilde{h}}\partial f)(B)\rangle.
\end{equation}
\rojito{Here $f(B)$ stands for the set of values of $f$ on the box $B$.}
Note that this condition is satisfied when either $B$ does not contain any zero of $f$ or no pair of gradient vectors of $f$ are orthogonal in $B$. 

In its abstract form, the Plantinga-Vegter algorithm is described in Algorithm~\nameref{alg:PVAlgorithm} below. The \textsc{Stan\-dard\-Subdivision} procedure in the description refers to taking a box $B$ and subdividing it into $2^n$ boxes of equal size.

\begin{algorithm2e*}[ht]
\DontPrintSemicolon
\SetKwInOut{input}{Input}
\SetKwInOut{output}{Output}
\caption{\textsc{PV-Abs\-tract}}\label{alg:PVAlgorithm}
\input{$f:\bbR^n\rightarrow \bbR$ with interval approximations $\square[hf]$ and $\square[\rojito{\tilde{h}}\nabla f]$\\
$a \in (0,\infty)$
}
\precondition{$V_\bbR(f)$ is smooth inside $[-a,a]^n$}
\hrulefill

$\tilde{\mcS}\leftarrow \{[-a,a]^n\}$\;
$\mcS\leftarrow\varnothing$\;
\Repeat{$\tilde{\mcS}=\varnothing$}{
Take $B$ in $\tilde{\mcS}$\;
$\tilde{\mcS}\leftarrow \tilde{\mcS}\setminus\{B\}$\;
\uIf{$C_f(B)$ {\tt true}}{
$\mcS\leftarrow \mcS\cup\{B\}$\;}
\Else{
$\tilde{\mcS}\leftarrow \tilde{\mcS}\cup \textsc{StandardSubdivision}(B)$\;
}}
\KwRet{$\mcS$}

\hrulefill

\output{Subdivision $\mcS\subseteq \square[-a,a]^n$ of $[-a,a]^n$}
\postcondition{For all $B\in\mcS$, $C_f(B)$ is true}
\end{algorithm2e*}

\subsection{Interval level: Algorithm~\texorpdfstring{\mdseries\nameref{alg:PVAlgorithmconcrete}}{\nameref{alg:PVAlgorithmconcrete}}}
To check condition $C_f(B)$, we use interval approximations allowing us to certify whether or not $0$ is in the image of $B$ under a certain map. Recall that an \emph{interval
approximation}~\cite{ratschek1984} of a function $F:\bbR^{m}\rightarrow \bbR^{m'}$ is a map
\[\square[F]:\square\bbR^{m}\rightarrow \square\bbR^{m'}\]
such that for all $B\in \square \bbR^{m} $,
\begin{equation}\label{eq:intervalcondition}
    F(B)\subseteq \square[F](B).
\end{equation}

A natural choice for the interval approximation of a $C^1$-function $F:\bbR^m\rightarrow \bbR^{m'}$ is its {\em standard interval approximation}
\[
\square\bbR^{m}\ni B\mapsto \square_{\mathrm{std}}[F](B)
:=F(m(B))+\sqrt{m}\left(\sup_{x\in B}\|\mathrm{D}_x F\|\right)
\left[-\frac{w(B)}{2},\frac{w(B)}{2}\right]^{m'}\]
where $\mathrm{D}_xF$ is the tangent map of $F$ at $x$ and $\|\mathrm{D}_xF\|$ its operator norm. Note that to construct this one in practice, we need to be able to evaluate $F$ and to compute efficiently upper bounds for $\sup_{x\in B}\|\mathrm{D}_xF\|$. In our case, this is possible due to the fact that we are working with polynomials. 

Let $f\in\Pd$. We will consider
\begin{equation}\label{eq:definitonhhprime}
    h(x)=\frac{1}{\|f\|(1+\|x\|^2)^{(d-1)/2}}
\quad\text{ and }\quad \rojito{\tilde{h}}(x)=\frac{1}{d\|f\|(1+\|x\|^2)^{d/2-1}}
\end{equation}
along with the maps
\begin{equation}\label{eq:fhat}
    \widehat{f}:x\mapsto h(x)f(x)=\frac{ f(x)}{\|f\|(1+\|x\|^2)^{(d-1)/2}}
\end{equation}
and
\begin{equation}\label{eq:derfhat}
\widehat{\partial f}:x\mapsto \rojito{\tilde{h}}(x)\partial f(x)=\frac{\partial f(x)}{d\|f\|(1+\|x\|^2)^{d/2-1}}
\end{equation}
where $\|f\|$ is the Weyl norm of $f$ (which we recall in Definition~\ref{defi:weylnorm}). 
\rojito{In~\textsection{\ref{sec:IA}} we will prove the following property of $\hat{f}$ and $\hat{\partial f}$.}

\begin{proposition}\label{prop:hats}
Let $f\in \Pd$. Then
\[\square[hf]:B\mapsto \widehat{f}(m(B))+(1+\sqrt{d})\sqrt{n}\left[-\frac{w(B)}{2},\frac{w(B)}{2}\right]
\]
is an interval approximation of $hf$, and
\[\square[\rojito{\tilde{h}}\partial f]:B\mapsto \widehat{\partial f}(m(B))+\big(1+\sqrt{d-1}\big)\sqrt{n}\left[-\frac{w(B)}{2},\frac{w(B)}{2}\right]^n\]
is an interval approximation of $\rojito{\tilde{h}}\partial f$.
\end{proposition}

\begin{remark}\label{rem:h}
\rojito{A natural question at this point is why we are using interval
approximations for $hf$ and $\tilde{h}\partial f$ instead  
of for $f$ and $\partial f$.
We work with  $hf$ and $\tilde{h}f$ for the sake of simplicity. We prefer to
work with the simpler interval approximations for $hf$ and
$\tilde{h}\partial f$ (shown in Proposition~\ref{prop:hats})
than with possibly more complex ones for $f$ and $\partial f$.}
\end{remark}

We now note that checking the condition ``$0\notin\langle B,
B\rangle$'' for a box $B$ can be reduced to checking
\[
  \sqrt{\frac{n}{2}}w(B)\leq \|m(B)\|.
\]
To do the latter we will use Lemma~\ref{lem:innerproductbound}
(which we also prove in \textsection{\ref{sec:IA}}). Together with
the interval approximations in Proposition~\ref{prop:hats}, we
derive a condition $C^{\square}_f$, implying $C_f(B)$ and easy 
to check.
\begin{theorem}\label{theo:condprime}
Let $B\in\square \bbR^n$. If the condition
\[
C_f^{\square}(B)\,:=\,
\abs{\widehat{f}(m(B))}>2\sqrt{dn}w(B)~\text{ or }~\left\|\widehat{\partial f}(m(B))\right\|>2\sqrt{2}\sqrt{d}nw(B).
\]
is satisfied, then $C_f(B)$ is true.
\end{theorem}
Theorem~\ref{theo:condprime} is the basis of the interval version of Algorithm~\nameref{alg:PVAlgorithmconcrete} below.
\begin{algorithm2e*}[ht]
\DontPrintSemicolon
\SetKwInOut{input}{Input}
\SetKwInOut{output}{Output}
\caption{\textsc{PV-In\-ter\-val}}\label{alg:PVAlgorithmconcrete}
\input{$f\in\Pd$\\
$a \in (0,\infty)$
}
\precondition{$V_\bbR(f)$ is smooth inside $[-a,a]^n$}
\hrulefill

$\tilde{\mcS}\leftarrow \{[-a,a]^n\}$\;
$\mcS\leftarrow\varnothing$\;
\Repeat{$\tilde{\mcS}=\varnothing$}{
Take $B$ in $\tilde{\mcS}$\;
$\tilde{\mcS}\leftarrow \tilde{\mcS}\setminus\{B\}$\;
\uIf{$\abs{\widehat{f}(m(B))}>(1+\sqrt{d})\sqrt{n}w(B)$}{
$\mcS\leftarrow \mcS\cup\{B\}$\;}
\ElseIf{$\left\|\widehat{\partial f}(m(B))\right\|>\sqrt{2}(1+\sqrt{d-1})nw(B)$}{
$\mcS\leftarrow \mcS\cup\{B\}$\;}
\Else{
$\tilde{\mcS}\leftarrow \tilde{\mcS}\cup \textsc{StandardSubdivision}(B)$\;
}}
\KwRet{$\mcS$}

\hrulefill

\output{Subdivision $\mcS\subseteq \square[-a,a]^n$ of $[-a,a]^n$}
\postcondition{For all $B\in\mcS$, $C_f(B)$ is true}
\end{algorithm2e*}

\begin{remark}\label{rem:BGT}
\rojito{There are other alternatives for interval
approximations 
and our framework has the flexibility to incorporate these
alternatives. For instance, the interval approximations
in~\cite{burr2017}, which we will  refer to as BGT, are based on
the Taylor expansion at the midpoint. In the interlude at the 
end of Section~\ref{sec:complexity}, we will show that our
complexity analysis also applies to this interval approximation.} 
\end{remark}

\begin{remark}\label{remark:intervalversion}
\rojito{We have described~\nameref{alg:PVAlgorithmconcrete}
without any reference to interval approximations. Such 
references have been replaced by explicit conditions on 
$\abs{\widehat{f}(m(B))}$ and 
$\left\|\widehat{\partial f}(m(B))\right\|$. }
\end{remark}

\subsection{Effective level: Algorithm~\texorpdfstring{\mdseries\nameref{alg:PVAlgorithmFP}}{\nameref{alg:PVAlgorithmFP}}}\label{sec:floating}

For the effective version (Algorithm~\nameref{alg:PVAlgorithmFP}), we will use floating-point numbers (cf.~\cite[\textsection {O.3.1}]{Condition} or~\cite[\textsection 1.2]{Higham96}). We do this, instead of using 
fixed-point or big rationals, because the use of  floating-point is computationally cheap, both in time and space. \rojito{We want to emphasize, however, that our use of floating-point numbers does not compromise the correctness of the algorithm (cf. Corollary~\ref{cor:correctness}).}

A floating-point number has the form
\[\pm\; 0.a_1a_2\cdots a_{\bfm}\; 2^{e}\]
where $a_1,\ldots,a_m\in\{0,1\}$ and $e\in\bbZ$. In general, the {\em number of significant digits}, $\bfm$, is fixed during the computation of arithmetic expressions, but it can be updated at different iterations of an algorithm if an increase in precision is needed.

We note that every real number $x\in \bbR$ has a floating-point approximation $r_{\bfm}(x)$ with $\bfm$ digits, such that
\[
  r_{\bfm}(x)=x(1+\delta)
\]
for some $\delta\in(-2^{-(\bfm-1)},2^{-(\bfm-1)})$. Moreover, given two floating-point numbers $x$ and $y$ with $\bfm$ significant digits, we can easily compute 
\[
  r_{\bfm}(x+y),\,r_{\bfm}(x-y),\,r_{\bfm}(xy),\,r_{\bfm}(x/y)\text{, and }r_{\bfm}\left(\sqrt{x}\right)
\]
in $\Oh(\bfm^2)$ bit-operations. Comparisons between floating-point numbers can also be made using this amount of 
bit-operations.

\begin{remark}
In the above estimation we are ignoring the complexity of adding the exponents or operating with them. In general the size of $e$ is of the order of $\abs{\log\abs{x}}$, and so the bit-size of $e$ is of the order of $\abs{\log\abs{\log\abs{x}}}$. This means that, unless the numbers we deal with are enormous, one should not worry about the bit-size of $e$ for cost estimates. 
\end{remark}

Finite-precision analyses do not rely on the precise 
form of floating-point numbers but just in some 
general properties which we now summarize. There 
is a subset $\bbF\subset\bbR$ of {\em floating-point numbers} (which we assume contains $0$), a {\em rounding map} 
$r:\bbR\rightarrow\bbF$, and a {\em round-off unit} 
$\bfu\in (0,1)$ satisfying the following conditions:
\begin{enumerate}
\item[(i)]
For any $x\in\bbF$, $r(x)=x$. In particular, $r(0)=0$.
\item[(ii)]
For any $x\in\bbR$, $r(x)=x(1+\delta)$ with $\abs{\delta}\leq \bfu$.
\end{enumerate}
Moreover, for $\circ\in \{+,-,\times,/\}$, there are approximate versions
\[\tilde{\circ}:\bbF\times \bbF\rightarrow \bbF\]
such that for all $x,y\in \bbF$,
\begin{equation}
    x\,\tilde{\circ}\,y=(x\circ y)(1+\delta)
\end{equation}
for some $\delta$ such that $\abs{\delta}<\bfu$. We also assume that there is
\[\widetilde{\sqrt{~}}:\bbF\rightarrow \bbF\]
such that for all $x\in\bbF$ with $x\geq 0$,
\[\widetilde{\sqrt{x}}=\sqrt{x}(1+\delta)\]
for some $\delta$ such that $\abs{\delta}<\bfu$. 
Each of these operations and comparisons between numbers in $\bbF$ can be done with cost $\Oh\left(\log^{2}\frac{1}{\bfu}\right)$. For the floating-point numbers we described above we have   $\bfu=2^{-(\bfm-1)}$. 

Once the way we deal with finite precision is clear, we introduce the efficient version of the Plantinga-Vegter algorithm (\nameref{alg:PVAlgorithmFP} below). We note that the algorithm updates the number of significant digits, $\bfm:=\abs{\log\bfu}+1$, depending on the width of the box that is being considered, being able, if necessary, to read the coefficients of $f$ with this updated precision. 

\smallskip

\begin{algorithm2e}[hbt]
\DontPrintSemicolon
\SetKwInput{input}{Input}
\SetKwInput{output}{Output}
\caption{\textsc{PV-Ef\-fec\-ti\-ve}}\label{alg:PVAlgorithmFP}

\input{$f\in\Pd$\\
$a \in [1,\infty)$
}
\precondition{$V_\bbR(f)$ is smooth inside $[-a,a]^n$}
\hrulefill

$\bfm_0\leftarrow 7+\left\lceil\log \sqrt{dn}\right\rceil$\;
$\tilde{\mcS}\leftarrow \{[-a,a]^n\}$\;
$\mcS\leftarrow\varnothing$\;
\Repeat{$\tilde{\mcS}=\varnothing$}{
Take $B$ in $\tilde{\mcS}$\;
$\tilde{\mcS}\leftarrow \tilde{\mcS}\setminus\{B\}$\;
$\bfm_B\leftarrow \bfm_0+\left\lceil\max\left\{\log a,\log(a/w(B))\right\}\right\rceil$\;
Switch to floating-point numbers with $\bfm_B$ significant digits\;
\uIf{$\abs{\widehat{f}(m(B))}>4\sqrt{dn}w(B)$}{
$\mcS\leftarrow \mcS\cup\{B\}$\;}
\ElseIf{$\left\|\widehat{\partial f}(m(B))\right\|>6\sqrt{d}nw(B)$}{
$\mcS\leftarrow \mcS\cup\{B\}$\;}
\Else{
$\tilde{\mcS}\leftarrow \tilde{\mcS}\cup \textsc{StandardSubdivision}(B)$\;
}}
\KwRet{$\mcS$}

\hrulefill

\output{Subdivision $\mcS\subseteq \square[-a,a]^n$ of $[-a,a]^n$}
\postcondition{For all $B\in\mcS$, $C_f(B)$ is true}
\end{algorithm2e}
\smallskip

\begin{remark}
\rojito{As in the case of~\nameref{alg:PVAlgorithmconcrete}, 
we could rewrite the conditions
$\abs{\widehat{f}(m(B))}>4\sqrt{dn}w(B)$ and
$\left\|\widehat{\partial f}(m(B))\right\|>6\sqrt{d}nw(B)$ 
in~\nameref{alg:PVAlgorithmFP} by 
$0\notin \tilde{\square}\left[hf\right]$ and 
$0\notin \tilde{\square}\left[\|\tilde{h}\partial f\|\right]$,
respectively, for some effective interval approximations
$\tilde{\square}$ (in the sense of~\cite{yap2019towards}).
Our writing of the algorithm, however, is led by the wish 
to explicitly describe the interval approximations we use, 
as noted in Remark~\ref{remark:intervalversion}.}
\end{remark}

\section{Main results}\label{sec:main}

In this section, we outline without proofs the 
main results of this paper. In the first part, 
we describe our randomness assumptions for polynomials. In the second one, we give precise statements for our bounds on the average and smoothed complexity of 
phase~I of the Plantinga-Vegter Algorithm with infinite precision. In the last part, we state similar results in the context of finite-precision arithmetic.

\subsection{Randomness Model}
Most of the literature on random multivariate polynomials considers polynomials with Gaussian independent coefficients and relies on techniques 
that 
are only useful for Gaussian measures. We will 
instead consider a general family of measures 
relying on robust techniques coming from 
geometric functional analysis. Let us recall 
some basic definitions. 

\begin{enumerate}
\item[(P1)] A random variable $\fkx\in\bbR$ is called \emph{centered} if $\bbE \fkx=0$. 
\item[(P2)] A random variable $\fkx\in\bbR$ is called \emph{subgaussian} if there exists a $K$ such that for all $p \geq 1$,
\[ (\bbE \abs{\fkx}^p)^{\frac{1}{p}} \leq K \sqrt{p}. \]
The smallest such $K$ is called the {\em $\Psi_2$-norm} 
of $\fkx$. 
\item[(P3)] A random variable $\fkx\in\bbR$ satisfies the 
\emph{anti-concentration property with constant $\rho$} if
\[ \max\left\{ \bbP \left( \abs{\fkx-u} \leq \eps \right) \mid u \in \bbR \right\}\leq \rho \eps .\]
\end{enumerate}

The subgaussian property (P2) has other equivalent formulations. We refer the interested reader to~\cite{V}. We note that the anti-concentration property (P3) is equivalent to having a density (with respect to the Lebesgue measure) bounded by $\rho/2$.
\begin{definition}
A \emph{dobro random polynomial} $\fkf\in \Hd$ with parameters $K$ and $\rho$ is a polynomial
\begin{equation} \label{randomdef}
\fkf:=\sum_{\abs{\alpha}=d}\binom{d}{\alpha}^{\frac{1}{2}}\fkc_\alpha X^\alpha
\end{equation}  
such that the $\fkc_\alpha$ are independent 
centered subgaussian random variables with $\Psi_2$-norm at most $K$
and anti-concentration property with constant $\rho$.
A \emph{dobro random polynomial} $\fkf\in \Pd$ is a polynomial $f$ such that its homogenization $\fkf\hm$ is so. 
\end{definition}
\begin{remark}
The word ``dobro'' comes from Russian and it means good. The word ``dobra'' in Turkish means straight and honest, and the word has similar connotations in Greek.
\end{remark}
Some \rojito{classes of} dobro random polynomials of interest are the following three.
\begin{enumerate}
    \item[(N)] A \emph{KSS random polynomial} is a dobro random polynomial such that each $\fkc_\alpha$ in \eqref{randomdef} is Gaussian with unit variance. For this model we have $K\rho = 1/\sqrt{2\pi}$.
    \item[(U)] A \emph{Weyl random polynomial} is a dobro random polynomial such that each $\fkc_\alpha$ in \eqref{randomdef} have uniform distribution in $[-1,1]$. For this model we have 
    $K\rho\leq 1$.
    \item[(E)] For $\ell\geq 2$, a \emph{$\ell$-random polynomial} is a dobro random polynomial whose coefficients are independent identically distributed random variables with density function
    \[t\mapsto \frac{1}{2\Gamma\left(1+\frac{1}{\ell}\right)}\,e^{-\abs{t}^{\ell}}.\]
    We have in this case that $\rho\leq 1$ and $K\leq 6/5$.
\end{enumerate}

\begin{remark}
The relevant complexity parameter for a dobro random polynomial $\fkf\in\Pd$ with constants $K$ and $\rho$ is the product $K\rho$. This is so because this product is invariant under scalings of $\fkf$ and condition numbers will be scale-invariant. Note that, \rojito{for $t>0$}, 
$t\fkf$ is still dobro, but with constants $tK$ and $\rho/t$.
\end{remark}

\begin{remark}\label{rem:unif}
If we are interested in integer polynomials, dobro
random polynomials may seem inadequate. One may be 
inclined to consider random polynomials $\fkf\in\Pd$ 
such that $\fkc_\alpha$ is a random integer in the
interval $[-2^{\tau},2^\tau]$, i.e., $\fkc_\alpha$ is 
a random integer of bit-size at most $\tau$. 

As $\tau\to\infty$ and after we normalize the 
coefficients dividing by $2^\tau$, this random 
model converges to that of Weyl random polynomials. 

Yet, in order to have a more satisfactory understanding of random integer polynomials, one has to consider  random variables without 
a continuous density function.  \rojito{The techniques we employed
in this note, coming originally from geometric functional
analysis, have already been used to analyze condition
numbers of random matrices with such discrete 
distributions~\cite{RV,V}.}
\end{remark}

\begin{remark}
\rojito{Even though there is a widespread agreement that average-case analysis is
a better picture of performance in practice than worst-case analysis, it
is not itself without contention. The most common objection to average-case
analysis is that its underlying probability distribution may not be an
accurate reflection of ``real life.'' In particular, that it may result in 
bounds that are too ``optimistic.''  
An alternate form of analysis, called
{\em smoothed analysis}, was introduced by D.~Spielmann and S.-H.~Teng
with the goal of overcoming this objection. The basic idea is to replace
``behavior at a random data'' by ``behavior at a random small perturbation
of arbitrary data.'' We won't attempt to describe the rationale of this
setting. This can be read in~\cite{ST:02,ST:09} or in~\cite[\textsection 2.2.7]{Condition}.
But as our development allows to include smoothed-analysis results without
a substantial additional effort, we do so in parts~(S) of
Theorems~\ref{thm:probMAIN},~\ref{thm:probMAINFP}, and~\ref{thm:BGT}.}
\end{remark}

\subsection{Complexity at the interval and effective levels}

The following two theorems give bounds for, respectively, the average and smoothed complexity of \nameref{alg:PVAlgorithmconcrete} and \nameref{alg:PVAlgorithmFP}.
\rojito{In both of them, the `big $\Oh$' notation is not
asymptotic. It refers to the existence of a multiplicative constant, which we don't specify, and holds for all values of
$a,K,\rho,d$ and $n$.}

\begin{theorem}\label{thm:probMAIN}\textsc{Complexity of }\nameref{alg:PVAlgorithmconcrete}:
\begin{enumerate}
    \item[(A)] Let $\fkf\in\Pd$ be a dobro random polynomial with parameters $K$ and $\rho$. The expected number of boxes in the final subdivision $\mcS$ of \nameref{alg:PVAlgorithmconcrete} on input $(\fkf,a)$ is at most
\[
d^{n}N^{\frac{n+1}{2}}\max\{1,a^n\} 2^{12n\log n + 8} (K\rho)^{n+1}
\]
and the expected number of arithmetic operations is at most
\[
\Oh\left(d^{n+1}N^{\frac{n+3}{2}}\max\{1,a^n\} 2^{12 n\log n +8}(K\rho)^{n+1}\right)
.\]
    \item[(S)] Let $f\in\Pd$, $\sigma>0$, and $\fkg\in\Pd$ a dobro random polynomial with parameters $K\geq 1$ and $\rho$ . Then the expected number of $n$-cubes of the final subdivision $\mcS$ of \nameref{alg:PVAlgorithmconcrete} on input $(\fkq_\sigma,a)$ where $\fkq_\sigma=f+\sigma\|f\|\fkg$ is at most
\[
d^{n}N^{\frac{n+1}{2}}\max\{1,a^n\} 2^{12 n\log n+8} (K\rho)^{n+1}\left(1+\frac{1}{\sigma}\right)^{n+1} 
\]
and the expected number of arithmetic operations is at most
\[
\Oh\left(d^{n+1}N^{\frac{n+3}{2}}\max\{1,a^n\} 2^{12 n\log n+8}
(K\rho)^{n+1}\left(1+\frac{1}{\sigma}\right)^{n+1}\right)
.\]
\end{enumerate}
\end{theorem}

\begin{theorem}\label{thm:probMAINFP}\textsc{Complexity of }\nameref{alg:PVAlgorithmFP}:
\begin{enumerate}
    \item[(A)] Let $\fkf\in\Pd$ be a dobro random polynomial with parameters $K$ and $\rho$. The expected number of boxes in the final subdivision $\mcS$ of \nameref{alg:PVAlgorithmFP} on input $(\fkf,a)$ is at most
\[
d^{n}N^{\frac{n+1}{2}}a^n 2^{15 n\log n+12} (K\rho)^{n+1}
\]
    and the expected number of arithmetic operations is at most
\[
\Oh\left(d^{n+1}N^{\frac{n+3}{2}}a^n 2^{15 n\log n+12} (K\rho)^{n+1}\right)
.\]
Moreover, the expected bit-cost of \nameref{alg:PVAlgorithmFP} on input $(\fkf,a)$ is at most
\[
  \Oh\left(d^{n+1}N^{\frac{n+3}{2}}a^n 2^{15n\log n+ 12}\log^2(dna) (K\rho)^{n+1}\right),
\]
under the assumptions that floating-point arithmetic is done using standard arithmetic and that the cost of operating with the exponents is negligible.
    \item[(S)] Let $f\in\Pd$, $\sigma>0$, and $\fkg\in\Pd$ a dobro random polynomial with parameters $K\geq 1$ and $\rho$ . Then the expected number of $n$-cubes of the final subdivision $\mcS$ of \nameref{alg:PVAlgorithmFP} on input $(\fkq_\sigma,a)$ where $\fkq_\sigma=f+\sigma\|f\|\fkg$ is at most
\[
d^{n}N^{\frac{n+1}{2}}a^n 2^{15 n\log n + 12} (K\rho)^{n+1}\left(1+\frac{1}{\sigma}\right)^{n+1}
\]
    and the expected number of arithmetic operations is at most
\[
\Oh\left(d^{n+1}N^{\frac{n+3}{2}}a^n 2^{15n\log n +12} (K\rho)^{n+1}\left(1+\frac{1}{\sigma}\right)^{n+1}\right)
.\]
Moreover, the expected bit-cost of \nameref{alg:PVAlgorithmFP} on input $(\fkq_\sigma,a)$ is at most
\[
  \Oh\left(d^{n+1}N^{\frac{n+3}{2}}a^n 2^{15 n\log n + 12}\log^2(dna) (K\rho)^{n+1}\left(1+\frac{1}{\sigma}\right)^{n+1}\right),
\]
under the assumptions that floating-point arithmetic is done using standard arithmetic and that the cost of operating with the exponents is negligible.
\end{enumerate}
\end{theorem}
\rojito{Fix a dimension $n$, a box $[-a,a]^n$ and a dobro distribution (and with it, the parameters $\rho$ and $K$).}
If $d$ is let to vary, $N=\binom{n+d}{n} 
\leq e^n (1+ \frac{d}{n})^n $. Hence the bounds of
Theorems~\ref{thm:probMAIN} and~\ref{thm:probMAINFP} are of 
the order $d^{\frac{n^2+5n}{2}}$. The complexity estimate  
in ~\cite[Theorem~4.3]{burr2017} reads as follows: 
\[ 2^{\Oh\left(  d^{n+1} (n\tau+n d\log{(nd)}) n \log{a} \right) } \]
with $\tau$ being the largest bit-size of the coefficients of $f$. One can see that the average analysis estimates 
(and the smoothed analysis, for a fixed $\sigma$) are exponentially smaller than this worst-case estimate. 
This seems to relate better with the efficiency in practice
of the Plantinga-Vegter algorithm.

We note, however, that the bound in~\cite{burr2017} and our bounds cannot be  directly compared. Not only because the former is 
worst-case and the latter average-case (or smoothed) but because 
of the different underlying settings: the bound in~\cite{burr2017} applies to  integer data, ours to real data. Nevertheless, the bounds for the effective version ~\nameref{alg:PVAlgorithmFP} apply to the real data under finite precision and provides estimates for the bit complexity.

\section{Geometric framework}\label{sec:geomfram}

There is an extensive literature on norms of
polynomials and their relation to norms of gradients
in $\Hd$. The PV algorithm, however, works in the affine space  with non-homogenous polynomials. We first establish basic definitions and inequalities that allow us to translate existing results into the setting of the PV algorithm. After the transfer is completed, we continue with establishing interval approximations. 

\subsection{Weyl norm}
We first introduce the Weyl \rojito{inner product} on $\Hd$.

\begin{definition}\label{defi:weylnormhom}
The \emph{Weyl inner product} on $\Hd$ is given by
\[\langle f,g\rangle:=\sum_{\alpha}\binom{d}{\alpha}^{-1}f_\alpha g_\alpha\]
for $f=\sum_\alpha f_\alpha X^\alpha,
g=\sum_\alpha g_\alpha X^\alpha,
\in \Hd$; and the \emph{Weyl inner product} on $\Hd^q$ is given by
\[
\langle\mathbf{f},\mathbf{g}\rangle
:=\sum_{i=1}^q\langle f_i,g_i\rangle
\]
for $\mathbf{f}=(f_i), \mathbf{g}=(g_i)\in \Hd^q$.
\end{definition}

To extend this inner product to $\Pd$, we use the homogeneization map
\begin{align*}
    \hm:\Pd&\rightarrow \Hd\\
    f&\mapsto f\hm:=f(X_1/X_0,\ldots,X_n/X_0)X_0^d.
\end{align*}
and its componentwise extension $\hm:\Pd^q\rightarrow \Hd^q$.

\begin{definition}\label{defi:weylnorm}
The \emph{Weyl inner product} on $\Pd^q$ is given by
\[\langle\mathbf{f},\mathbf{g}\rangle
:=\langle\mathbf{f}\hm,\mathbf{g}\hm\rangle\]
for $\mathbf{f},\mathbf{g}\in \Pd^q$.
\end{definition}
\rojito{For both $\Hd^q$ and $\Pd^q$ the Weyl norm is the norm induced by the Weyl
inner product}. 

Note that for $F\in \Hd^q$, we have that $\partial F(X)\in \mcH_{n,d-1}^{q(n+1)}$ and so we can talk about the Weyl norm of $\partial F(X)$. \rojito{Recall that we write explicitly the vector $X$ of indeterminates 
to indicate that we are working with $\partial F(X)$ as a vector of formal derivatives of $F$.} The following proposition comes in handy. 

\begin{proposition}\label{prop:ineqweylhom}
Let $\mathbf{f}\in \Hd^q$ and $y\in\bbS^n$. Then,
$(1) \; \|\mathbf{f}(y)\|\leq \|\mathbf{f}\|$ , $ (2) \; \left\|\diff_y \mathbf{f}_{\vert\Tg_y\bbS^n}\right\|\leq \sqrt{d}\|\mathbf{f}\|$ , 
$(3) \; \|\partial \mathbf{f}(X)\|\leq d\|\mathbf{f}\|$.
\end{proposition}

\begin{proof}
(1) is~\cite[Lemma~16.6]{Condition}, (2) the Exclusion Lemma~\cite[Lemma~19.22]{Condition}, and (3) by a direct computation, arguing as in the proof of~\cite[Lemma~16.46]{Condition}. Alternatively, one can also see~\cite[1\textsuperscript{\textsection 1}]{tonellicuetothesis} for a direct account of the proofs.
\end{proof}

\subsection{Central projection and homogeneization}

Let $\yuproj:\bbR^n\rightarrow \bbS^n$ be the map given by
\begin{equation}\label{eq:juproj}
    \yuproj:x\mapsto \frac{1}{\sqrt{1+\|x\|^2}}\begin{pmatrix}1\\ x\end{pmatrix}.
\end{equation}
One can see that $\yuproj$ is the map induced by the
central projection of \rojito{$\{1\}\times \bbR^n$} onto
the sphere $\bbS^n$ and that this map induces a
diffeomorphism between $\bbR^n$ and the upper half of
$\bbS^n$. 

Given $\mathbf{f}\in\Pd^q$, we observe that
\begin{equation}\label{eq:comphomphi}
\mathbf{f}\hm(\yuproj(x))=\frac{\mathbf{f}(x)}{(1+\|x\|^2)^{d/2}},
\end{equation}
and so, by the chain rule,
\begin{equation}\label{eq:chainrule}
\diff_{\yuproj(x)}\mathbf{f}\hm \Diff_x\yuproj=\frac{\diff_x\mathbf{f}}{(1+\|x\|^2)^{d/2}}
-\frac{d\cdot \mathbf{f}(x)x\transp}{(1+\|x\|^2)^{d/2+1}}
\end{equation}
\rojito{where $\diff_y\mathbf{f}\hm:\Tg_y\bbR^{n+1}\cong \bbR^{n+1}\rightarrow \Tg_{\mathbf{f}\hm(y)\bbR^{q}\cong}\bbR^q$,
$\diff_x\mathbf{f}:
\Tg_x\bbR^{n}\cong\bbR^n\rightarrow \Tg_{\mathbf{f}(x)}\bbR^{q}\cong\bbR^q$ and $\Diff_x\yuproj:\Tg_x\bbR^n\rightarrow
\Tg_{\yuproj(x)}\bbS^n = \yuproj(x)^{\perp}$ 
are respectively the tangent maps of  
$\mathbf{f}\hm$, $\mathbf{f}$ and~$\yuproj$.}

It is important to note that 
$\yuproj$ deforms the metric.
For each $x\in\bbR^n$, we can see that the singular values of $\diff_x\yuproj$ are
\begin{equation}
    \sigma_1\left(\diff_x\yuproj\right)=\cdots=\sigma_{n-1}\left(\diff_x\yuproj\right)=\frac{1}{\sqrt{1+\|x\|^2}},\,\sigma_n\left(\diff_x\yuproj\right)=\frac{1}{1+\|x\|^2},
\end{equation}
and so, in particular,
\begin{equation}\label{eq:boundsphi}
\|\Diff_x\yuproj\|=\frac{1}{\sqrt{1+\|x\|^2}}. 
\end{equation}

With the above, we next prove a version of
Proposition~\ref{prop:ineqweylhom} for $\Pd^q$.

\begin{proposition}\label{prop:lipschitz}
Let $\mathbf{f}\in \Pd^q$ be a polynomial map. 
Then the map 
\[
\mathbf{F}:x\mapsto\frac{ \mathbf{f}(x)}{\|\mathbf{f}\|(1+\|x\|^2)^{(d-1)/2}}
\]
is $(1+\sqrt{d})$-Lipschitz and, for all $x$, $\big\|\mathbf{F}(x)\big\|\leq \sqrt{1+\|x\|^2}$.
\end{proposition}
\begin{proof}
For the Lipschitz property, it is enough to bound the norm of the derivative of the map by $1+\sqrt{d}$. Due to~\eqref{eq:comphomphi},
\begin{equation}\label{eq:expr-F}
\mathbf{F}(x)=\sqrt{1+\|x\|^2}\,\frac{\mathbf{f}\hm(\yuproj(x))}{\|\mathbf{f}\|}
\end{equation}
and so, by the chain rule, 
\[
\diff_x\mathbf{F}=\frac{\mathbf{f}\hm(\yuproj(x))}{\|\mathbf{f}\|}
\frac{x\transp}{\sqrt{1+\|x\|^2}}+\sqrt{1+\|x\|^2}\,
\frac{\Diff_{\yuproj(x)} \mathbf{f}\,\Diff_x\yuproj}{\|\mathbf{f}\|}.
\]
Now, by the triangle inequality, 
\[
\|\diff_x\mathbf{F}\|\leq\frac{\|\mathbf{f}\hm(\yuproj(x))\|}{\|\mathbf{f}\|}
\frac{\|x\|}{\sqrt{1+\|x\|^2}}+\sqrt{1+\|x\|^2}\,
\frac{\|\Diff_{\yuproj(x)} \mathbf{f}\,\Diff_x\yuproj\|}{\|\mathbf{f}\|}.
\]
On the one hand,
\[
\frac{\|\mathbf{f}\hm(\yuproj(x))\|}{\|\mathbf{f}\|}\leq 1,
\]
by Proposition~\ref{prop:ineqweylhom}(1). On the other hand,
\begin{multline*}
\|\Diff_{\yuproj(x)} \mathbf{f}\,\Diff_x\yuproj\| = \left\|\Diff_{\yuproj(x)} \mathbf{f}_{\vert\Tg_{\yuproj(x)}\bbS^n}\,\Diff_x\yuproj \right\|\\\leq \left\|\Diff_{\yuproj(x)} \mathbf{f}_{\vert\Tg_{\yuproj(x)}\bbS^n}\right\|\|\Diff_x\yuproj \|\leq \frac{\sqrt{d}\|\mathbf{f}\|}{\sqrt{1+\|x\|^2}},
\end{multline*}
by Proposition~\ref{prop:ineqweylhom}(2) and~\eqref{eq:boundsphi}.
Hence
\[
\|\diff_x\mathbf{F}\|\leq \frac{\|x\|}{\sqrt{1+\|x\|^2}}+\sqrt{d}\leq 1+\sqrt{d}
\]
as we wanted to show.
The claim about $\big\|\mathbf{F}(x)\big\|$ follows from  Proposition~\ref{prop:ineqweylhom}(1) applied to \rojito{the expression~\eqref{eq:expr-F} for $\mathbf{F}$}.
\end{proof}

\subsection[Interval approximations]{Interval approximations}\label{sec:IA}

Recall that our interval approximations, given in Proposition~\ref{prop:hats}, rely on the functions 
$\widehat{f}$ and $\widehat{\partial f}$ given, respectively, in~\eqref{eq:fhat} and~\eqref{eq:derfhat}. The following lemma will give us the justification of our interval approximations, and with it a proof of Proposition~\ref{prop:hats}.

\begin{lemma}\label{lem:lipschitz}
Let $f\in\Pd$. Then:
\begin{enumerate}[(1)]
    \item The map $\widehat{f}$ given in~\eqref{eq:fhat} is $(1+\sqrt{d})$-Lipschitz and for all $x\in\bbR^n$, it satisfies $\abs{\widehat{f}(x)}\leq \sqrt{1+\|x\|^2}$.
    \item The map $\widehat{\partial f}$ given in~\eqref{eq:derfhat} is $(1+\sqrt{d-1})$-Lipschitz and for all $x\in\bbR^n$, it satisfies $\|\widehat{\partial f}(x)\|\leq \sqrt{1+\|x\|^2}$.
\end{enumerate}
\end{lemma}

\begin{proof}[Proof of Proposition~\ref{prop:hats}]
It is a straightforward consequence of the Lipschitz properties in Lemma~\ref{lem:lipschitz}. 
\end{proof}

\begin{proof}[Proof of Lemma~\ref{lem:lipschitz}]
(1) Apply Proposition~\ref{prop:lipschitz} with 
$\mathbf{f}=f$, then $\widehat{f}=\mathbf{F}$ and \rojito{both}
claims follow.

(2) Apply Proposition~\ref{prop:lipschitz} with $\mathbf{f}=f$, then
$\widehat{\partial f}=\frac{\|\partial f(X)\|}{d\|f\|}\mathbf{F}$
and the claims follow since $\frac{\|\partial f(X)\|}{d\|f\|}\leq 1$ 
by Proposition~\ref{prop:ineqweylhom} (3).
\end{proof}

Once we have shown that our interval approximations are so,
we show Theorem~\ref{theo:condprime} which reduces the 
interval condition $C_f(B)$ to the condition $C_f^{\square}(B)$ 
at a point.

\begin{lemma}\label{lem:innerproductbound}
Let $x\in \bbR^n$ and $s\in [0,1/\sqrt{2}]$. Then for all $v,w\in B(x,s\|x\|)$, we have $\langle v,w\rangle>\|v\|\|w\|(1-2s^2)\geq 0$.
\end{lemma}

\begin{proof}[Proof of Theorem~\ref{theo:condprime}]
By the \rojito{standard $\ell_2$-$\ell_\infty$} inequality\rojito{---which states that $\norm{x}\leq \sqrt{n}\norm{x}{_\infty}$ for $x\in\bbR^n$---},
interval approximations of Proposition~\ref{prop:hats} satisfy that for all $B\in\square\bbR^n$
\begin{equation}\label{eq:intcond1}
\mathrm{dist}\left((hf)(m(B)),\square[hf](B)\right)\leq \big(1+\sqrt{d}\big)\sqrt{n}\,w(B)/2
\end{equation}
and
\begin{equation}\label{eq:intcond2}
\mathrm{dist}\left((\rojito{\tilde{h}}\partial f)(m(\rojito{B})),\square[\rojito{\tilde{h}}\partial f](B)\right)\leq \big(1+\sqrt{d-1}\big)nw(B)/2
\end{equation}
where $\mathrm{dist}$ is the usual Euclidean distance.

When \rojito{the inequality on $\widehat{f}(m(B))$ 
in $C_f^{\square}(B)$} is satisfied,
then~\eqref{eq:intcond1} guarantees that $0\notin \square[hf](B)$. Similarly, when 
\rojito{the inequality on $\widehat{\partial f}(m(B))$ 
in $C_f^{\square}(B)$} is satisfied, 
then~\eqref{eq:intcond2} and Lemma~\ref{lem:innerproductbound} (with $s=1/\sqrt{2}$) guarantee that 
$0\notin \langle \square[\rojito{\tilde{h}}\partial f](B),\square[\rojito{\tilde{h}}\partial f](B)\rangle$. Hence $C_f^{\square}(B)$ implies $C_f(B)$. 
\end{proof}

\begin{proof}[Proof of Lemma~\ref{lem:innerproductbound}]
Let $s=\cos\,\theta$, so that $\theta\in [0,\pi/4]$, $c=\sqrt{1-s^2}$ and $K_c:=\{u\in\bbR^n\mid \langle x,u\rangle\geq \|x\|\|u\|c\}$ the convex cone of those vectors $u$ whose angle $\widehat{x\,u}$ with $x$, is at most~$\theta$.

Given $v,w\in K_c$, we have, by the triangle inequality, that $\angle(v\,w)\leq \angle(v\,x)+\angle(x\,w)\leq 2\theta\leq \pi/2$ \rojito{(here $\angle$ denotes angle)}. Thus
\[\cos\,\angle(v\,w)\geq \cos\left(\angle(v\,x)+\angle(x\,w)\right)\geq \cos\,2\theta=1-2s^2\geq 0.\]
And so, it is enough to show that 
$\rojito{B(x,\rojito{s}\|x\|)}\subseteq K_c$ or, equivalently, that 
$\mathrm{dist}(x,\partial K_c)\leq \rojito{s}\|x\|$.

Now, $\mathrm{dist}(x,\partial K_c)=\displaystyle \min\left\{\|x-u\|\mid\rojito{u\in K_c,\,\langle x,u\rangle=\|x\|\|u\|c}\right\}$ 
and this minimum equals 
the distance of $x$ to a line having an angle $\theta$ 
with $x$, which is $\|x\|s$.
\end{proof}

\section{Condition number}\label{sec:condition}

As other numerical algorithms in computational geometry, the
Plantinga-Vegter algorithm has a cost which significantly 
varies with inputs of \rojito{the same size, even if the
coefficients 
are rational and inputs have the same bit-size}. One wants to
explain this variation in terms of geometric properties of 
the input. Condition numbers allow for such an explanation.

\begin{definition}\cite{BCL17,Cu:99,CKS16}\label{def:condition}
Given $f\in\Hd$, \rojito{$f\ne 0$}, the \emph{local condition number} of $f$ at $y\in\bbS^n$ is
\[
\kappa(f,y):=\frac{\|f\|}{\sqrt{f(y)^2+\frac{1}{d}\|\diff_y f_{\vert\Tg_y\bbS^n}\|^2}}.
\]
Given $f\in\Pd$, the \emph{local affine condition number} of $f$ at $x\in\bbR^n$ is
\[
\kappaff(f,x):=\kappa(f\hm,\yuproj(x)).
\]
\end{definition}

\subsection[What does kappaff measure?]{What does $\kappaff$ measure?}

The nearer the hypersurface $V_\bbR(f)$ is to having 
a singularity at $x\in\bbR^n$,  the smaller are 
the boxes drawn by the Plantinga-Vegter algorithm around $x$. 
Instead of controlling how near $x$ is of being a singularity of $f$, we perform a Copernican turn and we control instead how near $f$ is of having a singularity at $x$. This is precisely what $\kappaff(f,x)$ does. 

\begin{theorem}[Condition Number Theorem]\label{cor:conditionumberthm}
Let $x\in\bbR^n$ and 
\begin{equation}\label{eq:sigmax}
    \Sigma_x:=\{g\in\Pd\mid g(x)=0,\,\diff_x g=0\}
\end{equation}
be the set of polynomials in $\Pd$ that have a singularity at $x$. Then for every $f\in\Pd$,
\[ 
\frac{\|f\|}{\kappaff(f,x)}=\mathrm{dist}
(f,\Sigma_x) 
\]
where $\mathrm{dist}$ is the distance induced by the Weyl norm 
on $\Pd$.
\end{theorem}

\begin{proof}
This is a reformulation of~\cite[Theorem~4.4]{BCL17} (cf.~\cite[Proposition~19.6]{Condition}).
\end{proof}

Theorem~\ref{cor:conditionumberthm} provides a
geometric interpretation of the local condition
number, and a corresponding ``intrinsic" complexity
parameter as desired by Burr, Gao and Tsigaridas
in~\cite{burr2017,burr2020}.  The next result is an
essential tool for our probabilistic analyses. Note that,
in the case under consideration, $\Sigma_x$ is a linear
subspace of codimension $n+1$ inside $\Pd$.
\begin{corollary}\label{cor:orthogonalprojection}
Let $x\in\bbR^n$ and let
$\fo_x:\Pd\rightarrow \Sigma_x^\perp$
be the orthogonal projection onto the orthogonal
complement of the linear subspace $\Sigma_x$. Then
\[
\kappaff(f,x)=\frac{\|f\|}{\|\fo_xf\|}.
\]
\end{corollary}
\begin{proof}
We have that $\mathrm{dist}(f,\Sigma_x)
=\|\fo_x f\|$ since $\Sigma_x$ is a linear subspace. 
Hence Theorem~\ref{cor:conditionumberthm} finishes 
the proof.
\end{proof}

\subsection{Regularity inequality}

After doing our Copernican turn, we can control how near is $f\in\Pd$ of having a singularity at $x\in\bbR^n$. The regularity inequality~\cite[Proposition~3.6]{BCTC1}  (cf.~\cite[Proposition~1\textsuperscript{\textsection 2}3]{tonellicuetothesis}) allows us to recover how near is $x$ of being a singularity of $f$. More precisely, the regularity inequality gives lower bounds for the value of the function or its derivative in terms of the condition number.

\begin{proposition}[Regularity inequality]\label{prop:fundamentalproposition_aff}
Let $f\in\Pd$ and $x\in\bbR^n$. Then either
\[\abs{\widehat{f}(x)}> \frac{1}{2\sqrt{2d}\,\kappaff(f,x)}  
\text{ or }
\left\|\widehat{\partial f}(x)\right\|> \frac{1}{2\sqrt{2d}\,\kappaff(f,x)}.\]
\end{proposition}

\begin{proof}
Without loss of generality assume that $\|f\|=1$. Let $y:=\yuproj(x)$, $g:=f\hm$ and assume that the first inequality does not hold. Then, by~\eqref{eq:comphomphi}, 
\[
  \abs{g(y)}\leq \frac{1}{2\sqrt{2d}\,\kappa(g,y) \sqrt{1+\|x\|^2}}.
\]

Now,
\[
\frac{1}{\sqrt{2}\kappa(g,y)}\leq \max\left\{\abs{g(y)},\frac{1}{\sqrt{d}}
\|\partial_yg\vert_{\Tg_y\bbS^n}\|\right\}
=\frac{1}{\sqrt{d}}\|\partial_yg\vert_{\Tg_y\bbS^n}\|,
\]
since $\abs{g(y)}<\frac{1}{\sqrt{2}\kappa(g,y)}$. Thus, by~\eqref{eq:chainrule} and ~\eqref{eq:boundsphi}, we get
\[\frac{1}{\sqrt{2}\,\kappa(g,y)}\leq \left\|\frac{\diff_xf}{(1+\|x\|^2)^{d/2}}-\frac{df(x)x^T}{(1+\|x\|^2)^{d/2+1}}\right\|\left(\frac{1+\|x\|^2}{\sqrt{d}}\right).\]

We divide by $\sqrt{d}$ and use the triangle inequality to obtain
\begin{align*}
    \frac{1}{\sqrt{2d}\,\kappa(g,y)}&\leq
    \frac{\|\diff_xf\|}{d(1+\|x\|^2)^{d/2-1}}
    +\frac{\abs{f(x)}}{(1+\|x\|^2)^{(d-1)/2}}
    \frac{\|x\|}{\sqrt{1+\|x\|^2}}\\
    &=\left\|\widehat{\partial f}(x)\right\|
    +\abs{\widehat{f}(x)}\frac{\|x\|}{\sqrt{1+\|x\|^2}}.
\end{align*}
By our assumption and $\|x\|< \sqrt{1+\|x\|^2}$, the above inequality implies
\[
\frac{1}{\sqrt{2d}\,\kappa(g,y)}<\left\|\widehat{\partial f}(x)\right\|+\frac{1}{2\sqrt{2d}\,\kappaff(f,x)},
\]
from where the desired inequality follows.
\end{proof}

\section{Complexity Analysis of the Interval version}\label{sec:complexity}
We analyze the complexity of~\nameref{alg:PVAlgorithmconcrete} in terms of the number of arithmetic operations the algorithm performs. This task reduces to estimating the number of boxes in the final subdivision produced by the algorithm. At the interval level, this is so, because each iteration of the algorithm takes the same number of arithmetic operations and the number of iterations is bounded by twice the number of final cubes. This was the underlying strategy in~\cite{burr2017}.

\subsection{Local size bound framework}

The original analysis in~\cite{burr2017} was based on the notion of local size bound.
\begin{definition}
A \emph{local size bound} for $C:\square\bbR^n\rightarrow \{{\tt True},{\tt False}\}$ is a function
$b:\bbR^n\rightarrow [0,\infty)$ such that for 
all $x\in\bbR^n$,
\[
b(x)\leq \inf\{\vol(B)\mid x\in B\in \square\bbR^n
\text{ and }C(B)\text{ }{\tt False}\}.
\] 
\end{definition}
The idea behind the local size bound is that it gives us the size from which every box containing $x$ satisfies $C$. In our case, we will apply this to the condition $C_f^{\square}$ introduced in Theorem~\ref{theo:condprime}. 

The following result, based on the notion of {\em continuous amortization} developed by Burr, Krahmer and
Yap~\cite{burr2009,burr2016} is proven
in~\cite[Proposition~5.2]{burr2017}.

\begin{theorem}\label{theo:analysis2}\cite{burr2009,burr2016,burr2017}
The number of boxes in the final subdivision 
$\mcS$ returned by~\nameref{alg:PVAlgorithmconcrete} on input $(f,a)$ is at most
\[
\max\left\{1,\int_{[-a,a]^n}\,\frac{2^n}{b(x)}\,
\mathrm{d} x\right\}
\]
where $b$ is a local size bound for $C_f^{\square}$ (of Theorem~\ref{theo:condprime}). Moreover, the bound is finite if and only if the algorithm terminates.\eproof
\end{theorem}

To effectively use Theorem~\ref{theo:analysis2} 
we need explicit constructions for the local 
size bound.

\subsection{Condition-based local size bound and complexity}

The following result expresses a local size bound for $C_f^{\square}$
in terms of the local condition number 
$\kappaff(f,x)$.

\begin{theorem}\label{thm:MAIN1}
The map
\[
  x\mapsto 1/\left(2^{5/2}dn\kappaff(f,x)\right)^n
\]
is a local size bound for $C_f^{\square}$ (of Theorem~\ref{theo:condprime}).
\end{theorem}

\begin{proof}
Let $x\in \bbR^n$. Since $x\in B$, $\|x-m(B)\|\leq \sqrt{n}w(B)/2$. Hence, by Lemma~\ref{lem:lipschitz} and the regularity inequality (Proposition~\ref{prop:fundamentalproposition_aff}), either
\[
\abs{\widehat{f}(m(B))}\geq \frac{1}{2\sqrt{2d}\,\kappaff(f,x)}
-(1+\sqrt{d})\sqrt{n}\,w(B)/2
\]
or
\[
\norm{\widehat{\partial f}(m(B))}\geq \frac{1}{2\sqrt{2d}\,\kappaff(f,x)}-(1+\sqrt{d-1})
\sqrt{n}\,w(B)/2.\]
This means that $C_f^{\square}(B)$ is true if either
\[
  2\sqrt{2d}\,(1+\sqrt{d})\sqrt{n}\,\kappaff(f,x)w(B)< 1
\]
or
\[
2\sqrt{2d}\,(1+\sqrt{d-1})n\kappaff(f,x)w(B)<1.
\]
Hence we get that $C_f^{\square}(B)$ is true when both conditions are satisfied and the inequality
$1+\sqrt{d}\leq 2\sqrt{d}$ finishes the proof.
\end{proof}

Using the results above, we get the following 
theorem exhibiting a condition-based 
complexity analysis of
Algorithm~\ref{alg:PVAlgorithm}.

\begin{theorem}\label{thm:MAIN2}
The number of boxes in the final subdivision $\mcS$
of~\nameref{alg:PVAlgorithmconcrete} on input $(f,a)$ 
is at most
\[d^n\max\{1,a^n\}2^{n\log{n}+\frac{9}{2}n}\,\bbE_{\fkx\in[-a,a]^n}\left(\kappaff(f,\fkx)^n\right).
\]
The number of arithmetic operations performed by~\nameref{alg:PVAlgorithmconcrete} on input $(f,a)$ is at most
\[
  \Oh\left(d^{n+1}\max\{1,a^n\}2^{n\log{n}+\frac{9}{2}n}N\,\bbE_{\fkx\in[-a,a]^n}
  \left(\kappaff(f,\fkx)^n\right)\right).
\]
\end{theorem}

\begin{proof}
\rojito{The first statement} follows from Theorems~\ref{theo:analysis2} and~\ref{thm:MAIN1} combined with the fact that $\int_{[-a,a]^n}\,\kappaff(f,x)^n\,\mathrm{d} x$ equals $(2a)^n\,\bbE_{\fkx\in[-a,a]^n}\left(\kappaff(f,\fkx)^n\right)$. The latter follows from the fact that one performs $\Oh(dN)$ arithmetic operations to test $C_f^{\square}$ and that the number of boxes that the algorithm generates is at most two times the number of final boxes.
\end{proof}

The above condition-based complexity estimate will become the main tool to prove  Theorem~\ref{thm:probMAIN} in Section~\ref{sec:probability} where we will study the quantity $\bbE_{\fkx\in[-a,a]^n}\left(\kappaff(f,\fkx)^n\right)$ 
for random $\fkf$.

In the literature on numerical algorithms in real algebraic
geometry~\cite{BCL17,BCTC1,BCTC2,CKMW1,CKMW2,CKMW3,CKS16},
it is customary the use the following global condition
number
\[
\kappaff(f):=\max_{x\in[-a,a]^n}\kappaff(f,x).
\]
The quantity $\bbE_{\fkx\in[-a,a]^n}
\left(\kappaff(f,\fkx)^n\right)$ in Theorem~\ref{thm:MAIN2}
is an average quantity, \rojito{whereas} the condition
number $ \kappaff(f) $ is a global supremum. The average
quantity has finite expectation (over $\fkf$), whereas the
global supremum does not admit a bounded first moment. This
shows that a condition-based precision control combined
with adaptive complexity techniques such as continuous
amortization  may lead to substantial improvements in
computational real algebraic geometry.  

\subsection{Interlude: Complexity of the interval version of~\texorpdfstring{~\cite{burr2017}}{Burr, Gao and Tsigarias (2017)}}

In~\cite{burr2017}, Burr, Gao and Tsigaridas gave an interval version of \nameref{alg:PVAlgorithm} different from \nameref{alg:PVAlgorithmconcrete} based in the BGT interval approximation which relies on Taylor series. We provide a condition-based and probabilistic complexity analysis of this algorithm, although only for the interval version, on which 
we only bound the number of cubes and not the number of arithmetic operations.

We recall that Burr, Gao and Tsigaridas~\cite{burr2017} showed that
 \begin{multline*}
        \mcC(f,x)  :=\min \left\{ \frac{2^{n-1}d/\ln{\left(1+2^{2-2n}\right)}+\sqrt{n}/2}{\mathrm{dist}(x,V_\bbC(f) )} ,\right.\\ \left. \frac{2^{2n}(d-1)/\ln{\left(1+2^{2-4n}\right)}+\sqrt{n/2}}{\mathrm{dist}( (x,x),V_\bbC(g_f))} \right\}
 \end{multline*}
where $g_f$ is the polynomial $\langle \diff
f(X),\partial f(Y)\rangle$, is a local size bound for the condition that their interval version of~\nameref{alg:PVAlgorithm} checks. 

\begin{theorem}\label{thm:burr2017}\cite{burr2017}
The map
\[
  x\mapsto 1/\mcC(f,x)^n
\]
is a local size bound function for the condition that the BGT interval version of\rojito{~\nameref{alg:PVAlgorithm}} checks.\eproof
\end{theorem}

Looking at the definition of $\mcC(f,x)$ in
\cite{burr2017} one can see that $1/\mcC$ measures 
how near is $x$ of being a singular zero of $f$. 
This is similar to $1/\kappaff$ which, by
Theorem~\ref{cor:conditionumberthm}, measures how 
near is $f$ of having $x$ as a singular zero. 
The following result relates these two quantities.

\begin{theorem}\label{thm:boundBGTbycond}
Let $d>1$ and $f\in\Pd$. Then, for all $x\in\bbR^n$,
\[
 \mcC(f,x)\leq 2^{3n}d^2\kappaff(f,x).
\]
\end{theorem}

\begin{proof}
Note that Lemma~\ref{lem:lipschitz} holds 
over the complex numbers as well. Due to this 
and the fact that $V_\bbC(f)=V_\bbC(\widehat{f})$,
we have that
\[
\abs{\widehat{f}(x)}\leq (1+\sqrt{d})\,\mathrm{dist}(x,V_\bbC(f)).
\]

Now, if $\sqrt{2}(1+\sqrt{d-1})\,\mathrm{dist}
((y_1,y_2),(x,x))<\|\widehat{\partial f}(x)\|$, 
then $\sqrt{2}(1+\sqrt{d-1})\|y_i-x\|<\|
\widehat{\partial f}(x)\|$. Thus, by
Lemma~\ref{lem:lipschitz},
$\sqrt{2}\left\|\widehat{\partial f}(y_i)-
\widehat{\partial f}(x)\right\|<\left\|\widehat{\partial f}(x)\right\|$ 
and so, by Lemma~\ref{lem:innerproductbound}, 
$0\neq \langle \widehat{\partial f}(y_1), 
\widehat{\partial f}(y_2)\rangle$. Hence
\[
  \left\|\widehat{\partial f}(x)\right\|\leq \sqrt{2}(1+\sqrt{d-1})\,
  \mathrm{dist}(x,V_\bbC(g_f)).
\]
The bound now follows from
Proposition~\ref{prop:fundamentalproposition_aff},
together with $2^{3(n-1)}d+\sqrt{n}\leq 2^{3n-2}d$ 
and
\begin{equation*}
\min\left\{\frac{2^{n-1}d}
{\ln{\left(1+2^{2-2n}\right)}}
+\frac{\sqrt{n}}{2},\frac{2^{2n}(d-1)}
{\ln{\left(1+2^{2-4n}\right)}}+\sqrt{\frac{n}{2}}\right\}\leq 2^{3n-4}d+\frac{\sqrt{n}}{2}.
\end{equation*}
The latter follows from
\begin{equation}
\frac{1}{\ln{\left(1+2^{2-2n}\right)}}\leq 2^{2n-3}
\text{~ and ~}
\frac{1}{\ln{\left(1+2^{2-4n}\right)}}\leq 2^{4n-3},
\end{equation}
which are deduced from first-order approximations of the natural logarithm.
\end{proof}

\rojito{Theorems~\ref{thm:burr2017} and~\ref{thm:boundBGTbycond} 
combine to give an analog of Theorem~\ref{thm:MAIN1} for the 
BGT interval version of~\nameref{alg:PVAlgorithm}. 
Also,~\cite[Theorem~5.1]{burr2017} provides an analog 
of Theorem~\ref{theo:analysis2} in this setting. We can 
therefore proceed to derive the following result, a BGT 
version of Theorem~\ref{thm:MAIN2}, in the same manner that 
the latter is derived from Theorems~\ref{theo:analysis2} and~\ref{thm:MAIN1}.}

\begin{corollary}\label{cor:condcompBGT}
The number of boxes in the final subdivision $\mcS$
of the BGT interval version of
Algorithm~\nameref{alg:PVAlgorithm} on input $(f,a)$ is at 
most
\begin{equation}\tag*{\qed}
  d^{2n}\max\{1,a^n\}2^{3n^2+2n}\,\bbE_{x\in[-a,a]^n}
  \left(\kappaff(f,x)^n\right).
\end{equation}
\end{corollary}

\begin{remark}
The main difference between $C(f,x)$ and 
$\kappa(f,x)$ is that $C(f,x)$ is a non-linear
quantity and is hard to compute and to analyze, while the local
condition number $\kappa(f,x)$---as indicated in Corollary \ref{cor:orthogonalprojection}---is a 
linear quantity, easier to compute and analyze. 
\end{remark}

We finish this interlude giving \rojito{a form of  
Theorem~\ref{thm:probMAIN} for the BGT version of 
\nameref{alg:PVAlgorithm} (which, obviously, deals only with number 
of boxes, not with number of arithmetic operations). It is proved 
as Theorem~\ref{thm:probMAIN} 
(see~\textsection{\ref{sec:avg-cost}} 
and~\textsection{\ref{sec:smoothed-cost}}) with 
Corollary~\ref{cor:condcompBGT} taking the role of 
Theorem~\ref{theo:analysis2}}.

\begin{theorem}\label{thm:BGT}\ 
\begin{enumerate}
    \item[(A)] Let $\fkf\in\Pd$ be a dobro random polynomial with parameters $K$ and $\rho$.
The expected number of boxes in the final subdivision $\mcS$ of the BGT interval version of \nameref{alg:PVAlgorithm} on input $(\fkf,a)$ is at most
\[ d^{n^2}N^{\frac{n+1}{2}}\max\{1,a^n\} 2^{3n^2+n\log n+7n+\frac{15}{2}}  (K\rho)^{n+1}.\]
    \item[(S)] Let $f\in\Pd$, $\sigma>0$, and $\fkg\in\Pd$ a dobro random polynomial with parameters $K\geq 1$ and $\rho$ . Then the expected number of boxes of the final subdivision $\mcF$ of the BGT interval version of \nameref{alg:PVAlgorithm} on input $(\fkq_\sigma,a)$ where $\fkq_\sigma=f+\sigma\|f\|\fkg$ is at most
\begin{equation}\tag*{\qed}
    d^{n^2}N^{\frac{n+1}{2}}\max\{1,a^n\} 2^{3n^2+n\log n+7n+\frac{15}{2}}  (K\rho)^{n+1} 
    \rojito{\left(1+\frac1\sigma\right)^{n+1}}.
\end{equation}
\end{enumerate}
\end{theorem}

\section{Error and complexity analysis of the effective version}\label{sec:FP}

\rojito{We next work on the framework of floating-point numbers introduced in \textsection{\ref{sec:floating}}.}
For an arithmetic expression $\phi$ and a point $x\in\bbR$,
we will denote by $\fl(\phi(x))\in\bbF$ the value obtained
when evaluating $\phi$ at $r(x)\in\bbF$ using
floating-point finite precision. In general, our objective
is to show that for such expressions $\phi$ in our
algorithm we have, for some other expression $\psi(x)$ and
some \rojito{$k\geq 1$} satisfying $k\bfu<1$, 
\[\fl(\phi(x))=\phi(x)+\psi(x)\theta_k\]
where $\theta_k$ is any number $\delta\in\bbR$ 
satisfying
\[
\abs{\delta}\leq \frac{k\bfu}{1-k\bfu}.
\]
This is the general strategy in~\cite[Chapter~3]{Higham96}.

\subsection{Finite-precision computations}

We study the errors due to finite-precision in
algorithm~\nameref{alg:PVAlgorithmFP} and show its
correctness. \rojito{In all 
what follows, we use {\em numerical algorithm} to refer to 
an algorithm meant to be implemented with finite precision 
and analyzed in terms of error accumulation. This is common 
terminology.}

The following two propositions \rojito{bound the forward
error in the computation of $\lvert\hat{f}(x)\,\rvert$ and
$\|\widehat{\partial f}(x)\|$. Because their proofs are a
variation of well-known results
(e.g.~\cite[Thm.~6.10]{CKMW1}) and are more tedious than
enlightening, we defer them to an appendix.}

\begin{proposition}
\label{prop:error-absolutevalueofevaluation}
There is a numerical algorithm which, with input $f\in\Pd$
and $x\in \bbR^n$, computes $\lvert\hat{f}(x)\rvert$. 
This algorithm performs $\Oh(dN)$ arithmetic operations,
and, on input $x\in\bbF^n$ and
$f\in\Pd\cap\bbF[X_1,\ldots,X_n]$, the computed value
$\fl(\lvert\hat{f}(x)\rvert)$ satisfies
\[\fl(\lvert\hat{f}(x)\rvert)=\lvert\hat{f}(x)\rvert
+\sqrt{1+\|x\|}\errsymb{32d\log(n+1)}.\]
In particular, if the round-off unit satisfies
\[
\bfu\leq \frac{1}{64d\log(n+1)},
\]
then for $x\in[-a,a]^n\cap\bbF^{n}$, 
\[\lvert\fl(\lvert\hat{f}(x)\rvert)
-\lvert\hat{f}(x)\rvert\rvert
\leq 64\sqrt{2}d\sqrt{n+1}\log(n+1)\max\{1,a\}\bfu.\]
The above remains true for arbitrary $f$ and $x$ if we apply the algorithm to $r(f)$ and $r(x)$.
\end{proposition}

\begin{proposition}\label{prop:error-normgradientvector}
There is a numerical algorithm which, with input $f\in\Pd$
and $x\in \bbR^n$, computes $\|\hat{\partial f}(x)\|$. It
performs $\Oh(dN)$ arithmetic operations, and, on input
$x\in\bbF^n$ and $f\in\Pd\cap\bbF[X_1,\ldots,X_n]$, the
computed value $\|\hat{\partial f}(x)\|$ satisfies
\[\fl(\|\hat{\partial f}(x)\|)=\|\hat{\partial
f}(x)\|+\sqrt{1+\|x\|}\errsymb{32d\log(n+1)}.
\]
In particular, if the round-off unit satisfies
\[
\bfu\leq \frac{1}{64d\log(n+1)},
\]
then for $x\in[-a,a]^n\cap\bbF^{n}$, 
\[\lvert\fl(\|\hat{\partial f}(x)\|)
-\|\hat{\partial f}(x)\|\rvert\leq
64\sqrt{2}d\sqrt{n+1}\log(n+1)\max\{1,a\}\bfu.\]
The above remains true for arbitrary $f$ and $x$ if we
apply the algorithm to $r(f)$ and $r(x)$.
\end{proposition}

We can now show the correctness of
Algorithm~\nameref{alg:PVAlgorithmFP}. We will denote by
$\fl(B)$ the rounding $r(B)$ of a box $B$ given by
\[m(\fl(B))=m(B)(1+\errsymb{1})\text{\quad and\quad
}w(\fl(B))=w(B)(1+\errsymb{1}).
\]
Similarly, we will write $\fl(f)$ to denote the 
rounding $r(f)$ of $f$.
The next theorem shows that if the round-off unit is sufficiently small, then a floating-point version of condition $C_f^{\square}(B)$ is good enough to check $C_f(B)$.

\begin{theorem}\label{thm:condFP}
Let $B\in\square[-a,a]^n$. If
\[
C_f^{\mathrm{FP}}\,:=\,\left\{\begin{array}{rl}
&\fl\left(\lvert\widehat{\fl(f)}(m(\fl(B)))\rvert\right)
>\fl\left(4\sqrt{d}\sqrt{n+1}w(\fl(B))\right)\\[6pt]
\text{or }&\left(\|\widehat{\partial\fl(f)}
(m(\fl(B)))\|\right)
>\fl\left(6\sqrt{d}(n+1)w(\fl(B))\right)\end{array}\right.
\]
and
\[
\bfu\leq \frac{1}{128\sqrt{dn}}\frac{\min\{1,w(B)\}}{\max\{1,a\}},
\]
then $C_f^{\square}(B)$ holds and, hence, so does $C_f(B)$.
\end{theorem}

\begin{corollary}\label{cor:correctness}
Algorithm~\nameref{alg:PVAlgorithmFP} is correct. \eproof
\end{corollary}

\begin{proof}[Proof of Theorem~\ref{thm:condFP}]
Note that the conditions of
Propositions~\ref{prop:error-absolutevalueofevaluation}
and~\ref{prop:error-normgradientvector} are satisfied. 
Therefore, \rojito{using our hypothesis on the magnitude of $\bfu$, we have}
\begin{equation}\label{eq:evineq1}
  \abs{\hat{f}(m(B))}>\fl\left(\abs{\widehat{\fl(f)}(m(\fl(B)))}\right)-\sqrt{d}\log(n+1)\min\{1,w(B)\}  
\end{equation}
and that
\begin{equation}\label{eq:gradineq1}
\left\|\widehat{\partial f}(m(B))\right\|>\fl\left(\left\|\widehat{\partial\fl(f)}(m(\fl(B)))\right\|\right)-\sqrt{d}\log(n+1)\min\{1,w(B)\}.
\end{equation}

By error analysis (Proposition~\ref{prop:Higham}), we have that
\begin{equation}\label{eq:evineq2}
\fl\left(4\sqrt{d}\sqrt{n+1}w(\fl(B))\right)=4\sqrt{d}\sqrt{n+1}w(B)(1+\errsymb{8})
\end{equation}
and
\begin{equation}\label{eq:gradineq2}
\fl\left(4\sqrt{d}(n+1)w(\fl(B))\right)=6\sqrt{d}(n+1)w(B)(1+\errsymb{8}).
\end{equation}
Hence, \rojito{ again by the bound on $\bfu$, from~\eqref{eq:evineq2} we get}
\begin{equation}\label{eq:evineq3}
\fl\left(4\sqrt{d}\sqrt{n+1}w(\fl(B))\right)>4\sqrt{d}\sqrt{n+1}w(B)\left(1-\frac{1}{8\sqrt{dn}}\frac{\min\{1,w(B)\}}{\max\{1,a\}}\right)
\end{equation}
and \rojito{from \eqref{eq:gradineq2}
\begin{equation}\label{eq:gradineq3}
\fl\left(4\sqrt{d}(n+1)w(\fl(B))\right)>6\sqrt{d}(n+1)w(B)\left(1-\frac{1}{8\sqrt{dn}}\frac{\min\{1,w(B)\}}{\max\{1,a\}}\right)
\end{equation}
Now, combining~\eqref{eq:evineq1} and~\eqref{eq:evineq3}, we get}
\begin{align}\label{eq:multi1}
    \abs{\hat{f}(m(B))}&>2\sqrt{d}\sqrt{n+1}w(B)\\
    +2\sqrt{d}&\sqrt{n+1}w(B)\left(1-\frac{1}{4\sqrt{dn}}
    \frac{\min\{1,w(B)\}}{\max\{1,a\}}-\frac{\log(n+1)}{2\sqrt{n+1}}\min\left\{1,\frac{1}{w(B)}\right\}\right)\nonumber
\end{align}
\rojito{and, combining \eqref{eq:gradineq1} and \eqref{eq:gradineq3},}
\begin{align}\label{eq:multi2}
    \left\|\widehat{\partial f}(m(B))\right\|&
    >3\sqrt{d}(n+1)w(B)\\
    +3\sqrt{d}(n&+1)w(B)\left(1-\frac{1}{6\sqrt{dn}}\frac{\min\{1,w(B)\}}{\max\{1,a\}}-\frac{\log(n+1)}{2(n+1)}\min\left\{1,\frac{1}{w(B)}\right\}\right)\nonumber
\end{align}
Now, the term between parentheses in the right-hand side 
of~\eqref{eq:multi1} is positive since
\begin{multline*}
    \frac{1}{4\sqrt{dn}}\frac{\min\{1,w(B)\}}{\max\{1,a\}}+\frac{\log(n+1)}{2\sqrt{n+1}}\min\left\{1,\frac{1}{w(B)}\right\}\\\leq \frac{1}{4\sqrt{dn}}+\frac{\log(n+1)}{2\sqrt{n+1}}\leq \frac{1}{4}+\frac{1}{2}<1,
\end{multline*}
and so is the one in the right-hand side 
of~\eqref{eq:multi2} since
\begin{multline*}
\frac{1}{6\sqrt{dn}}\frac{\min\{1,w(B)\}}{\max\{1,a\}}+\frac{\log(n+1)}{2(n+1)}\min\left\{1,\frac{1}{w(B)}\right\}\\\leq \frac{1}{6\sqrt{dn}}+\frac{1}{2\sqrt{n+1}}\leq \frac{1}{6}+\frac{1}{2\sqrt{2}}<1.
\end{multline*}
Therefore our claim holds.
\end{proof}

\subsection{Complexity of Algorithm~\nameref{alg:PVAlgorithmFP}} \label{sec:FPA}

We now prove the analogous of Theorem~\ref{thm:MAIN1} in the finite-precision setting. To do so we have to slightly 
modify the sense of the term `local size bound' to take 
finite precision into account. 

\begin{definition}
A \emph{local size bound} for $C_f^{\mathrm{FP}}$ is a function
$b_f^{\mathrm{FP}}:\bbR^n\rightarrow [0,\infty)$ such that for 
all $x\in\bbR^n$,
\[
b_f^{\mathrm{FP}}(x)\leq \inf\left\{\vol(B)\,\Big\vert\, \begin{array}{rl}x\in B\in \square\bbR^n
&\text{, }C_f^{\mathrm{FP}}(B)\text{ }{\tt False}\\&\text{ with }\bfu\leq \frac{1}{128\sqrt{dn}}\frac{\min\{1,w(B)\}}{\max\{1,a\}}\end{array}\right\}.
\] 
\end{definition}

The modifications takes into account that the condition $C_f^{\mathrm{FP}}$ is checked with sufficiently large precision, as indicated by Theorem~\ref{thm:condFP}. The theorem below gives us the local size bound for finite precision. 

\begin{theorem}\label{thm:MAIN1FP}
The map
\[x\mapsto 1/\left(2^6dn\kappaff(f,x)\right)^n\]
is a local size bound for $C_f^{\mathrm{FP}}$ (of Theorem~\ref{thm:condFP}).
\end{theorem}
\begin{proof}
The proof is similar to the one of~Theorem~\ref{thm:condFP}. For now on, let $B\in\square \bbR^n$ be such that $x\in B$. 

By Proposition~\ref{prop:error-absolutevalueofevaluation} and~\ref{prop:error-normgradientvector} and the bound on $\bfu$, we have that
\[
\fl\left(\abs{\widehat{\fl(f)}(m(\fl(B)))}\right)>\abs{\hat{f}(m(B))}-\sqrt{d}\log(n+1)\min\{1,w(B)\}
\]
and that
\[
\fl\left(\left\|\widehat{\partial\fl(f)}(m(\fl(B)))\right\|\right)>\left\|\widehat{\partial f}(m(B))\right\|-\sqrt{d}\log(n+1)\min\{1,w(B)\}.
\]
By error analysis (Proposition~\ref{prop:Higham}),
\[
4\sqrt{d}\sqrt{n+1}w(B)\left(1+\frac{1}{8\sqrt{dn}}\frac{\min\{1,w(B)\}}{\max\{1,a\}}\right)>\fl\left(4\sqrt{d}\sqrt{n+1}w(\fl(B))\right)
\]
and
\[
6\sqrt{d}(n+1)w(B)\left(1+\frac{1}{8\sqrt{dn}}\frac{\min\{1,w(B)\}}{\max\{1,a\}}\right)>\fl\left(4\sqrt{d}(n+1)w(\fl(B))\right).
\]

By the regularity inequality (Proposition~\ref{prop:fundamentalproposition_aff}) and Corollary~\ref{lem:lipschitz}, we know that either
\begin{align*}
  \fl&\left(\abs{\widehat{\fl(f)}(m(\fl(B)))}\right)\\&>\frac{1}{2\sqrt{2d}\kappaff(f,x)}-\frac{(1+\sqrt{d})\sqrt{n}}{2}w(B)-\sqrt{d}\log(n+1)\min\{1,w(B)\}\\
  &>\frac{1}{2\sqrt{2d}\kappaff(f,x)}-2\sqrt{dn}w(B)
\end{align*}
or
\begin{align*}
  \fl&\left(\left\|\widehat{\partial\fl(f)}(m(\fl(B)))\right\|\right)\\&>\frac{1}{2\sqrt{2d}\kappaff(f,x)}-\frac{(1+\sqrt{d-1})\sqrt{n}}{2}w(B)-\sqrt{d}\log(n+1)\min\{1,w(B)\}\\
  &>\frac{1}{2\sqrt{2d}\kappaff(f,x)}-2\sqrt{dn}w(B).
\end{align*}
Hence $C_f^{\mathrm{FP}}(B)$ holds as long as
\[
\frac{1}{2\sqrt{2d}\kappaff(f,x)}-2\sqrt{dn}w(B)>6\sqrt{d}(n+1)w(B)\left(1+\frac{1}{8\sqrt{dn}}\frac{\min\{1,w(B)\}}{\max\{1,a\}}\right),
\]
which is implied by
\[
2^6d(n+1)\kappaff(f,x)w(B)<1.
\]
This means that $C_f^{\mathrm{FP}}(B)$ is true when $\vol(B)<1/\left(2^6dn\kappaff(f,x)\right)^n$, which is what we wanted to show.
\end{proof}

Using continuous amortization~\cite{burr2009,burr2016} 
(we use the statement in~\cite[Theorem~5]{burr2020}), 
we obtain the following condition-based complexity 
analysis of~\nameref{alg:PVAlgorithmFP}.

\begin{theorem}\label{thm:MAIN2FP}
The number of boxes in the final subdivision $\mcS$
of~\nameref{alg:PVAlgorithmFP} on input $(f,a)$ 
is at most
\[d^na^n2^{n\log{n}+8n}\,\bbE_{\fkx\in[-a,a]^n}\left(\kappaff(f,\fkx)^n\right).
\]
The number of arithmetic operations performed by~\nameref{alg:PVAlgorithmFP} on input $(f,a)$ is at most
\[
  \Oh\left(d^{n+1}a^n2^{n\log{n}+8n}N\,\bbE_{\fkx\in[-a,a]^n}
  \left(\kappaff(f,\fkx)^n\right)\right).
\]
Furthermore, the bit-cost of \nameref{alg:PVAlgorithmFP} on input $(f,a)$ is at most
\[
  \Oh\left(d^{n+1}a^n2^{n\log{n}+8n}N\log^2(dna)\,\bbE_{\fkx\in[-a,a]^n}
  \left(\kappaff(f,\fkx)^n\log^2\kappaff(f,x)\right)\right)
\]
under the assumptions that floating-point arithmetic is done using standard arithmetic and that the cost of operating with the exponents is negligible.
\end{theorem}

\begin{proof}
The first two claims follow from 
Theorems~\ref{thm:MAIN1FP} 
and~\ref{theo:analysis2}. For the third claim, we recall
the following variant of Theorem~\ref{theo:analysis2} that
can be found in~\cite[Theorem~5]{burr2020}. Let $\mcS$ be
the final subdivision output
by~\nameref{alg:PVAlgorithmconcrete} and
$h:(0,\infty)\rightarrow (0,\infty)$ a continuous map. 
Then
\[
\sum_{B\in\mcS}h\left(w(B)\right)\leq
\max\left\{h(2a),\int_{[-a,a]^n}\,\frac{2^n}
{b_f^{\mathrm{FP}}(x)}\,
h\left(\frac{b_f^{\mathrm{FP}}(x)^{\frac{1}{n}}}{2}\right)
\,\mathrm{d}x\right\}.
\]
Applying Theorem~\ref{thm:MAIN1FP}, \rojito{we get that $\sum_{B\in\mcS}h\left(w(B)\right)$ is bounded by
\[
\max\left\{h(2a),2^{n\log n+7n}d^n \int_{[-a,a]^n}\,\,\kappaff(f,x)^nh\left(2^5dn
\kappaff(f,x)\right)\,\mathrm{d}x\right\}.
\]}
Now, we note that testing $C_f^{\mathrm{FP}}$ at each of
the boxes along the way takes at most $\Oh(dN)$ arithmetic
operations and that the number of boxes that the algorithm
deals with is at most twice the number of final boxes.
Because of this, the bit-cost of the algorithm (ignoring
the cost of operating with exponents) in floating-point
arithmetic is
\[
\Oh\left(dN\sum_{B\in\mcS}\bfm_B^2\right).
\]
This is so, because each arithmetic operation takes
$\Oh(\bfm^2)$ bit-time and $\bfm_B$ is the largest  
precision needed to test $C_f^{\mathrm{FP}}$ in any box 
that is an ancestor of $B$. Hence, \rojito{by Theorem ~\ref{thm:condFP} and the relation of $\bfm_B$ to $\bfu$, }taking
\[
h(w(B))=\Oh\left(\max\left\{\log^2 2^9\sqrt{dn}a,\log^2 2^9\sqrt{dn}\frac{a}{w(B)}\right\}\right)
\]
gives the final bound.
\end{proof}

The above condition-based complexity estimate will become the complexity estimates in Theorem~\ref{thm:probMAINFP} in the coming Section~\ref{sec:probability}.

\section{Probabilistic analyses}\label{sec:probability}

In this section, we prove Theorems~\ref{thm:probMAIN} 
and~\ref{thm:probMAINFP} stated in Section~\ref{sec:main}
using  Theorems~\ref{thm:MAIN1} and~\ref{thm:MAIN1FP} \rojito{and their corollaries}
respectively. 

\subsection{\rojito{Some useful tools}}

The main tools we are going to use are  a tail bound 
on the 
norm of a random vector and a small ball type estimate to 
ensure norm of a random projection is not too small.
Following~\cite[5\textsuperscript{\textsection
1}]{tonellicuetothesis}, we will give explicit constants 
avoiding the use of undefined absolute constants. 
This will require us to sketch some proofs.

\begin{theorem}\label{thm:probtool1} 
Let $\fkx\in\bbR^N$ be a random vector where each 
component $\fkx_i$ is centered and sub-Gaussian with
\rojito{$\Psi_2$}-norm $K$. Then for all 
$t \geq 5K\sqrt{N}$, 
\begin{equation}
 \bbP \left( \norm{\fkx} \geq t\right)  \leq  \exp\left(-\frac{t^2}{(5K)^2}\right) .
 \end{equation}
\end{theorem}

\begin{proof}[Sketch of proof] We follow the ideas in~\cite[Theorems~2.6.3]{V}. Note that $\|\fkx\|\geq t$ is equivalent to $e^{s^2\|\fkx\|^2}\geq e^{s^2t^2}$. By Markov's inequality and independence,
\[\bbP\left(\|\fkx\|\geq t\right)\leq e^{-s^2t^2}\bbE e^{s^2\|\fkx\|^2}=\prod_{i=1}^N\bbE e^{s^2\fkx_i^2}.\]
By assumption, for each $i$,
\[
\bbE e^{s^{2l}\fkx_i^2}=\sum_{l=0}^\infty \frac{s^{2l}\bbE\fkx_i^{2l}}{l!}\leq \sum_{l=1}^\infty \frac{s^{2l}K^{2l}(2l)^l}{l!}\leq \sum_{l=0}^\infty \left(2eK^2s^l\right)^{l},
\]
since $l!\geq (l/e)^l$. Thus, taking $s^2=1/(4eK^2)$, 
we get
\[\bbP\left(\|\fkx\|\geq t\right)=2^N e^{-t^2/(4eK^2)}.\]
The claim is now trivial assuming 
$t\geq \sqrt{8e\ln(2)}K\sqrt{N}$.
\end{proof}

\begin{theorem}\cite[Corollary~1.4]{RV-1}
\label{thm:probtool2}
Let $\fkx\in\bbR^N$ be a random vector where each component
$\fkx_i$ has the anti-concentration property with constant
$\rho$ and $P:\bbR^N\rightarrow\bbR^N$ an orthogonal
projection onto a $k$-dimensional linear subspace of
$\bbR^N$. Then for all $\varepsilon>0$,
\[
\bbP\left(\|P\fkx\|\leq \sqrt{k}\varepsilon \right)\leq \left(3\rho\varepsilon\right)^k.
\]
\end{theorem}

\begin{proof}[Sketch of proof]
Note that by assumption, each $\fkx_i$ has probability
density (with respect to the Lebesgue measure) bounded by
$\rho/2$. Then, by~\cite[Theorem 1.1.]{grigoris16}, 
$P\fkx$ has probability density (with respect to the
Lebesgue measure) bounded by
$\left(\rho/\sqrt{2}\right)^k$. Thus
\[\bbP\left(\|P\fkx\|\leq \sqrt{k}\varepsilon \right)\leq
\omega_k\left(\frac{\sqrt{k}\rho}{\sqrt{2}}\right)^k\]
where $\omega_k$ is the volume of the $k$-dimensional
Euclidean ball.

Now, $\omega_kk^{\frac{k}{2}}\leq
(2e)^{\frac{k}{2}}\pi^{\frac{k}{2}}$, 
from where the claim follows.
\end{proof}

\subsection{Average Complexity Analysis}\label{sec:avg-cost}

The following theorem is the main technical result
from which the average complexity bound will follow. 

\begin{theorem}\label{thm:boundlocalcondition}
Let $\fkf\in\Pd$ be a dobro random polynomial with parameters $K$ and $\rho$. For all $x\in\bbR^n$ 
and $t\geq e$,
\[
\bbP \left(\kappaff(\fkf,x) \geq t\right) \leq 
 2 \left(\frac{N}
 {n+1}\right)^{\frac{n+1}{2}}(15 K\rho)^{n+1}
 \frac{\ln(t)^{\frac{n+1}{2}}}{t^{n+1}} .
 \]
\end{theorem}

\begin{remark}\label{remark:boundconstants}
By~\cite[(1)]{EPR18}, we have 
$K \rho \geq \frac{1}{4}$ for a dobro random polynomial $\fkf$ 
with parameters $K$ and $\rho$. This fact will be used without mention in the bounds below.
\end{remark}

\begin{proof}[Proof of Theorem~\ref{thm:boundlocalcondition}] By Corollary~\ref{cor:orthogonalprojection}, we have that
$\kappaff(\fkf,x)=\|\fkf\|/\|\fo_x\fkf\|$ with $\fo_x$ an orthogonal projection onto the $(n+1)$-dimensional linear subspace $\Sigma_x^{\perp}$. 

By the union bound, for all $u,t>0$,
\begin{equation}\label{eq:unionbound}
   \bbP \left( \kappaff(\fkf,x) \geq t \right)  \leq 
   \bbP \left( \norm{\fkf} \geq u \right) + \bbP \left( \norm{\fo_x \fkf} \leq u/t \right) .  
\end{equation}
We apply now Theorems~\ref{thm:probtool1} to the first term  and~\ref{thm:probtool2} to the second. Thus for 
$u>5K\sqrt{N}$ and $t>0$,
\[ 
 \bbP (\kappaff(\fkf,x) \geq t) \leq 
 \exp(-u^2/(5K)^2) + \left(\frac{3u \rho}
 {t\sqrt{n+1}}\right)^{n+1}.
\]
We set $u=5K\sqrt{N\ln(t)}$, so we get
\[ 
 \bbP \left(\kappaff(\fkf,x) \geq t\right) \leq 
 t^{-N}+ \left(\frac{15 K\rho\sqrt{N}}
 {\sqrt{n+1}}\right)^{n+1}
 \frac{\ln(t)^{\frac{n+1}{2}}}{t^{n+1}} 
\]
for $t\geq e$. The inequality $n+1 \leq N$ \rojito{and Remark~\ref{remark:boundconstants} finish} the proof.
\end{proof}

Theorem~\ref{thm:boundlocalcondition} immediately gives probabilistic bounds for the expressions  $\bbE_{\fkx\in[-a,a]^n}\left(\kappaff(\fkf,\fkx)^n\right)$ and  $\bbE_{\fkx\in[-a,a]^n}\left(\kappaff(\fkf,\fkx)^n\log^2\kappaff(\fkf,\fkx)\right)$ for a random $\fkf$. The two corollaries below, together with Theorems~\ref{thm:MAIN1} and~\ref{thm:MAIN1FP}, give 
us the proof of the part (A) of Theorems~\ref{thm:probMAIN} 
and~\ref{thm:probMAINFP}.

\begin{theorem}\label{thm:main3}
Let $\fkf\in\Pd$ be a dobro random polynomial with parameters $K$ and $\rho$ and $\alpha\in [1,n+1)$. Then
\[
\bbE_\fkf\bbE_{\fkx\in [-a,a]^n}\left(\kappaff(\fkf,\fkx)^{\alpha}\right)
\leq  4 \frac{\alpha\sqrt{n+1}}{n+1-\alpha}\left(\frac{N}
 {n+1-\alpha}\right)^{\frac{n+1}{2}}(25 K\rho)^{n+1}.
\]
\end{theorem}

\begin{corollary}\label{cor:main3}
Let $\fkf\in\Pd$ be a dobro random polynomial with parameters $K$ and $\rho$. Then
\[
\bbE_\fkf\bbE_{\fkx\in [-a,a]^n}\left(\kappaff(\fkf,\fkx)^{n}\right)
\leq  N^{\frac{n+1}{2}} 2^{5n+\frac{3}{2}\log n+\frac{15}{2}}
(K\rho)^{n+1}.
\]
\end{corollary}

\begin{corollary}\label{cor:main3FP}
Let $\fkf\in\Pd$ be a dobro random polynomial with parameters $K$ and $\rho$. Then
\[
\bbE_\fkf\bbE_{\fkx\in [-a,a]^n}\left(\kappaff(\fkf,\fkx)^{n}\log^2\kappaff(\fkf,\fkx)\right)
\leq  N^{\frac{n+1}{2}} 2^{6n+\frac{3}{2}\log n+12}
(K\rho)^{n+1}.
\]
\end{corollary}

\begin{proof}[Proof of Theorem~\ref{thm:main3}]
By the Fubini-Tonelli theorem,
\[
 \bbE_{\fkf}\bbE_{\fkx\in [-a,a]^n}\left(\kappaff(\fkf,\fkx)^\alpha\right)
 =\bbE_{\fkx\in [-a,a]^\alpha}\bbE_{\fkf}\left(\kappaff(\fkf,\fkx)^n\right)
\]
so it is enough to have a uniform bound for
\[
\bbE_{\fkf}\left(\kappaff(\fkf,x)^\alpha\right)=\int_1^\infty
\bbP\left(\kappaff(\fkf,x)^\alpha\geq t\right)\,\mathrm{d} t.
\]
Now, by Theorem~\ref{thm:boundlocalcondition}, this 
is bounded by
\[
e^\alpha+2 \left(\frac{N}
 {\alpha(n+1)}\right)^{\frac{n+1}{2}}(15 K\rho)^{n+1}
\int_1^\infty\,
\frac{\ln(t)^{\frac{n+1}{2}}}{t^{\frac{n+1}{\alpha}}}  
\, \mathrm{d} t.
\]
After the change of variables $t=e^{\frac{\alpha}{n+1-\alpha} s}$ the 
bound becomes
\begin{multline*}
    e^\alpha+2 \frac{\alpha}{n+1-\alpha}\left(\frac{N}
 {(n+1-\alpha)(n+1)}\right)^{\frac{n+1}{2}}(15 K\rho)^{n+1}
\int_1^\infty\,
s^{\frac{n+1}{2}}e^{-s} 
\, \mathrm{d} s\\
=e^\alpha+2 \frac{\alpha}{n+1-\alpha}\left(\frac{N}
 {(n+1-\alpha)(n+1)}\right)^{\frac{n+1}{2}}\Gamma\left(\frac{n+3}{2}\right)(15 K\rho)^{n+1},
\end{multline*}
where $\Gamma$ is Euler's Gamma function. We note that $e^\alpha\leq e^{n+1}$ and that, by the Stirling estimates,
\[ 
 \Gamma\left(\frac{n+3}{2}\right) \leq 
 \sqrt{2\pi} \left(\frac{n+3}{2e}\right)
 ^{\frac{n+2}{2}}\leq  \sqrt{2\pi}
 \left(\frac{n+1}{e}\right)^{\frac{n+2}{2}}.
\]
Combining all these 
inequalities, we obtain the desired upper bound.
\end{proof}

\begin{proof}[Proof of Corollary~\ref{cor:main3}]
We take $\alpha=n$ in Theorem~\ref{thm:main3}.
\end{proof}

\begin{proof}[Proof of Corollary~\ref{cor:main3FP}]
Recall that $\log^2y\leq 5\sqrt{y}$ for $y\geq 1$. Hence
\[\bbE_\fkf\bbE_{\fkx\in [-a,a]^n}\left(\kappaff(\fkf,\fkx)^{n}\log^2\kappaff(\fkf,\fkx)\right)\leq 2^{5/2}\bbE_\fkf\bbE_{\fkx\in [-a,a]^n}\left(\kappaff(\fkf,\fkx)^{n+\frac{1}{2}}\right)\]
and the claim follows using Theorem~\ref{thm:main3} with $\alpha=n+\frac{1}{2}$.
\end{proof}

\rojito{We can finally prove the average complexity bounds in 
our main theorems.}

\begin{proof}[Proof of Theorem~\ref{thm:probMAIN}(A)]
\rojito{The expected number of boxes we want to bound is 
bounded by the expectation of the estimate for this 
quantity in Theorem~\ref{thm:MAIN2} with respect to a dobro 
random $\fkf\in\Pd$, 
that is,
$$
d^n\max\{1,a^n\}2^{n\log{n}+\frac{9}{2}n}\,
\bbE_{\fkf\in\Pd}\bbE_{\fkx\in [-a,a]^n}\left(\kappaff(\fkf,\fkx)^{n}\right).
$$
A bound for the inner double expectation is 
in Corollary~\ref{cor:main3}. 
}

\rojito{The bound for the expected number of operations is
similarly derived.}
\end{proof}

\begin{proof}[Proof of Theorem~\ref{thm:probMAINFP}(A)]
\rojito{Similar to the proof above but using Corollaries~\ref{cor:main3}
and~\ref{cor:main3FP} to get upper bounds for the two 
expectations (arithmetic cost and, also now, bit-cost).}
\end{proof}

\subsection{Smoothed Complexity Analysis}\label{sec:smoothed-cost}

The tools used for our average complexity 
analysis yield also a smoothed complexity 
analysis (see~\cite{ST:02} 
or~\cite[\textsection 2.2.7]{Condition}). We provide this 
analysis following the lines of~\cite{EPR19}. 

The main idea of smoothed complexity is to have a
complexity measure interpolating between 
worst-case complexity and average-case complexity. 
More precisely, we are interested in the 
maximum ---over $f\in\Pd$--- of the average 
cost of the algorithm when the input polynomial 
has the form 
\begin{equation}\label{eq:perturbed}
    \fkq_{\sigma}:=f+\sigma \|f\|\fkg
\end{equation}
with $\fkg\in\Pd$ a dobro random polynomials with parameters $K\geq 1$ and $\rho$, and $\sigma\in(0,\infty)$. Notice that 
the perturbation $\sigma\|f\|\fkg$ of $f$ 
is proportional to both $\sigma$ and $\|f\|$.

The following lemma shows how Theorems~\ref{thm:probtool1} and~\ref{thm:probtool2} apply to this class of random polynomials.

\begin{lemma}\label{lem:smooth}
Let $\fkq_\sigma$ be as in~\eqref{eq:perturbed}. Then 
for $t > 1+\sigma\sqrt{N}$
\[ 
\bbP \left( \norm{\fkq_\sigma} \geq t \norm{f} \right) \leq \exp\left(-\frac{(t-1)^2}{\left(\sigma 5 K\right)^2}\right)
\]
and for every $x\in\bbR^n$,
\[ 
\bbP \left( \norm{\fo_x\fkq_\sigma} \leq \varepsilon \right) \leq \left(3 \rho \varepsilon /\left(\sigma\|f\|\sqrt{n+1}\right)\right)^{n+1}
\]
where $\fo_x$ is as in Corollary~\ref{cor:orthogonalprojection}.
\eproof
\end{lemma}

\begin{proof}
By the triangle inequality we have 
$\bbP(\norm{\fkq_{\sigma}} \geq t \norm{f}) 
\leq \bbP(\norm{\fkg} \geq (t-1) / \sigma)$. 
Then we apply Theorem~\ref{thm:probtool1} 
which finishes the proof of the first claim.
The second claim is a direct consequence of 
Theorem~\ref{thm:probtool2}. 
\end{proof}

As in the average case, this leads to a tail bound.

\begin{theorem} \label{thm:tailboundsmooth}
Let $\fkq_\sigma$ be as in~\eqref{eq:perturbed} and $x\in\bbR^n$. Then 
for $\sigma>0$ and $t\geq e$, 
\[
\bbP \left(\kappaff(\fkq_\sigma,x) \geq t\right) \leq 
 2 \left(\frac{N}
 {n+1}\right)^{\frac{n+1}{2}}(15 K\rho)^{n+1}
 \frac{\ln(t)^{\frac{n+1}{2}}}{t^{n+1}}\left(1+\frac{1}{\sigma}\right)^{n+1}.
\]
\end{theorem}

\begin{proof}
We proceed as in the proof of Theorem~\ref{thm:boundlocalcondition}, but with Lemma~\ref{lem:smooth} using $u=\|f\|(\sigma 5K\sqrt{N\ln(t)}+1)$. This gives the desired 
bound arguing as in that proof after noticing that
\[
  u\leq \|f\|(1+\sigma) 5K\sqrt{N\ln(t)}
\]
which holds since $5K\sqrt{N\ln(t)}\geq 1$. 
\end{proof}

As in the average case, Theorem~\ref{thm:tailboundsmooth} 
yields probabilistic bounds for both $\bbE_{\fkx\in[-a,a]^n}\left(\kappaff(\fkf,\fkx)^n\right)$ and $\bbE_{\fkx\in[-a,a]^n}\left(\kappaff(\fkf,\fkx)^n\log^2\kappaff(\fkf,\fkx)\right)$ for random $\fkf$. The two corollaries below, together with Theorems~\ref{thm:MAIN1} and~\ref{thm:MAIN1FP}, give us the proof of the part (S) of Theorems~\ref{thm:probMAIN} 
and~\ref{thm:probMAINFP}.

\begin{theorem}\label{thm:main3smoothed}
Let $\fkq_\sigma$ be as in~\eqref{eq:perturbed} and $\alpha\in[1,n+1)$. 
Then for all $f\in\Pd$ and all $\sigma>0$, 
\begin{multline*}
\bbE_{\fkq_\sigma}\bbE_{\fkx\in [-a,a]^n}\left(\kappaff(\fkq_\sigma,\fkx)^{\alpha}\right)
\\\leq  4 \frac{\alpha\sqrt{n+1}}{n+1-\alpha}\left(\frac{N}
 {n+1-\alpha}\right)^{\frac{n+1}{2}}(25 K\rho)^{n+1}\left(1+\frac{1}{\sigma}\right)^{n+1}.
\end{multline*}
\end{theorem}

\begin{proof}
The proof is as that of Theorem~\ref{thm:main3}, but
using Theorem~\ref{thm:tailboundsmooth} instead of
Theorem~\ref{thm:boundlocalcondition}. 
\end{proof}

\begin{corollary}\label{cor:main3smoothed}
Let $\fkq_\sigma$ be as in~\eqref{eq:perturbed}. Then for all 
$f\in\Pd$ and all $\sigma>0$,
\[
\bbE_{\fkq_\sigma}\bbE_{\fkx\in [-a,a]^n}\left(\kappaff(\fkq_\sigma,\fkx)^{n}\right)
\leq  N^{\frac{n+1}{2}} 2^{5n+\frac{3}{2}\log n+\frac{15}{2}}
(K\rho)^{n+1}\left(1+\frac{1}{\sigma}\right)^{n+1}.
\]
\end{corollary}

\begin{corollary}\label{cor:main3FPsmoothed}
Let $\fkq_\sigma$ be as in~\eqref{eq:perturbed}. Then for all $f\in\Pd$ and all $\sigma>0$,
\begin{multline*}
\bbE_{\fkq_\sigma}\bbE_{\fkx\in [-a,a]^n}\left(\kappaff(\fkq_\sigma,\fkx)^{n}\log^2\kappaff(\fkq_\sigma,\fkx)\right)\\
\leq  N^{\frac{n+1}{2}} 2^{6n+\frac{3}{2}\log n+12}
(K\rho)^{n+1}\left(1+\frac{1}{\sigma}\right)^{n+1}.
\end{multline*}
\end{corollary}
\begin{proof}[Proof of Corollaries~\ref{cor:main3smoothed} and~\ref{cor:main3FPsmoothed}]
We do as in the proof of Corollaries~\ref{cor:main3} and \ref{cor:main3FP} but using Theorem~\ref{thm:main3smoothed} instead of Theorem~\ref{thm:main3}.
\end{proof}

\rojito{
We conclude showing how the smoothed complexity estimates follow.}

\begin{proof}[Proof of Theorem~\ref{thm:probMAIN}(S)]
\rojito{The proof is the same as that of Theorem~\ref{thm:probMAIN}(A), but 
using Corollary~\ref{cor:main3smoothed} instead of
Corollary~\ref{cor:main3}.}
\end{proof}

\begin{proof}[Proof of Theorem~\ref{thm:probMAINFP}(S)]
\rojito{The proof is the same as that of Theorem~\ref{thm:probMAIN}(A), but 
using Corollaries~\ref{cor:main3smoothed} and~\ref{cor:main3FPsmoothed}
instead of Corollaries~\ref{cor:main3} and~\ref{cor:main3FP}.}
\end{proof}

\section*{Acknowledgements}
We cordially thank Michael Burr and Elias Tsigaridas for useful discussions. We also thank the two anonymous reviewers for 
their very detailed feedback that greatly helped 
us to improve this paper. 

Additionally, J. T.-C. is grateful to Evgenia Lagoda for moral support and Gato Suchen for useful suggestions for this paper.

\section*{Declarations: Funding}

This work was supported by the Einstein Fundation Berlin. F.C.~was partially supported by a GRF grant from the Research Grants Council of the Hong Kong SAR (project number CityU 11302418). A.E.~is supported by US National Science Foundation grant CCF 2110075. J. T.-C.~was by a postdoctoral fellowship of the 2020 ``Interaction'' program of the \emph{Fondation Sciences Mathématiques de Paris}, and partially supported by ANR JCJC GALOP (ANR-17-CE40-0009), the PGMO grant ALMA, and the PHC GRAPE. 

\section*{Declarations: Sources}
An extended abstract containing some of the results was presented at ISSAC'19~\cite{ISSAC}. Some preliminary versions of the results in Section~\ref{sec:probability} was included in the doctoral thesis of J. Tonelli-Cueto~\cite{tonellicuetothesis}.

\begin{appendices}
\section{Proofs of
Propositions~\ref{prop:error-absolutevalueofevaluation}
and~\ref{prop:error-normgradientvector}}\label{sec:technicalproofs}

We proceed by introducing a new error symbol which will make our manipulations easier, then we recall some fundamental numerical algorithms for computing inner product and monomials and we apply them to the computed quantities during the execution of algorithm~\nameref{alg:PVAlgorithmFP}.

\subsection{The arithmetic of error accumulation}

\rojito{To ease the technique of~\cite[Chapter~3]{Higham96}, we will use the symbol
$\theta_k$ allowing any real number $k\geq1$ in the subindex. Note that this does not affect any of the results.}

\rojito{As the symbol $\theta_k$ might be difficult to parse, let us explain in more detail how it works.} Let $\phi$ be some arithmetic expression. Whenever we write an expression of the form
\begin{equation}\label{def:error}
\fl(\phi(x))=\tilde{\phi}\left(x,\errsymb{t_1},\ldots,\errsymb{t_\ell}\right)
\end{equation}
for some arithmetic expression $\tilde{\phi}$ and for some real numbers $t_1,\ldots,t_\ell\geq 1$, we will mean that, as long
as $\max\{t_1,\ldots,t_\ell\}\bfu<1/2$, we have 
\[
\fl(\phi(x))=\tilde{\phi}\left(x,\tau_1,\ldots,\tau_\ell\right)
\]
for some
\[\tau_1\in \left[-\frac{t_1\bfu}{1-t_1\bfu},\frac{t_1\bfu}{1-t_1\bfu}\right],\ldots,\tau_\ell\in \left[-\frac{t_\ell\bfu}{1-t_\ell\bfu},\frac{t_\ell\bfu}{1-t_\ell\bfu}\right].
\]
We note that in this notation we are allowing more freedom 
as we don't require $t_1,\ldots,t_\ell$ to be integers.
Furthermore, and this will make it computationally as 
useful as Landau notation, we introduce the following 
additional, asymmetric, notation.

\rojito{Assume $\max\{t_1,\ldots,t_\ell,t'_1,\ldots,t'_{\ell'}\}\bfu<1/2$ and $x\in\bbR$.} We write 
\begin{equation}\label{def:equiv}
 \tilde{\phi}\left(x,\errsymb{t_1},\ldots,\errsymb{t_\ell}\right)
 =\tilde{\phi'}\left(x,\errsymb{t'_1},\ldots,\errsymb{t'_{\ell'}}\right)
\end{equation}
to mean that \rojito{for every}
\[\tau_1\in \left[-\frac{t_1\bfu}{1-t_1\bfu},\frac{t_1\bfu}{1-t_1\bfu}\right],\ldots,\tau_\ell\in \left[-\frac{t_\ell\bfu}{1-t_\ell\bfu},\frac{t_\ell\bfu}{1-t_\ell\bfu}\right],\]
there exist
\[\tau_1'\in \left[-\frac{t_1'\bfu}{1-t_1'\bfu},\frac{t_1'\bfu}{1-t_1'\bfu}\right],\ldots,\tau'_{\ell'}\in \left[-\frac{t'_{\ell'}\bfu}{1-t'_{\ell'}\bfu},\frac{t'_{\ell'}\bfu}{1-t'_{\ell'}\bfu}\right]\]
\rojito{---of course, depending on $\tau_1,\ldots,\tau_\ell$---} such that
\[
\tilde{\phi}\left(x,\tau_1,\ldots,\tau_\ell\right)
=\tilde{\phi'}\left(x,\tau'_1,\ldots,\tau'_{\ell'}\right).
\]
This is consistent with notation~\eqref{def:error} in the sense that if both~\eqref{def:error} 
and~\eqref{def:equiv} hold then 
$\fl(\phi(x))=\tilde{\phi'}\left(x,\errsymb{t_1},\ldots,\errsymb{t_\ell}\right)$. 
This will allow us to mechanically perform the finite precision analysis using the following rules.

\begin{proposition}\label{prop:Higham}
For all $s,s'\geq 1$, the following holds for the error symbol:
\begin{enumerate}[(E1)]
    \item[(E1)] If $s\leq s'$, $\errsymb{s}=\errsymb{s'}$.
    \item[(E2)] $\errsymb{s}+\errsymb{s'}+\errsymb{s}\errsymb{s'}=\errsymb{s+s'}$.\\In particular, $\errsymb{s}+\errsymb{s'}=\errsymb{s+s'}$ and  $(1+\errsymb{s})(1+\errsymb{s'})=1+\errsymb{s+s'}$.
    \item[(E3)] $(1+\errsymb{s})^{-1}=1+\errsymb{2s}$.
    \item[(E4)] $\sqrt{1+\errsymb{s}}=1+\errsymb{s}$.
    \item[(E5)] For all $t\in\bbR$, $t\errsymb{s}=\abs{t}\errsymb{s}=\errsymb{\max\{1,\abs{t}\}s}$.
    \item[(E6)] For all $t,t'\in\bbR$, $t\errsymb{s}+t'\errsymb{s'}=(\abs{t}+\abs{t'})\errsymb{\max\{s,y\}}$.
    \item[(E7)] For all $t,t'\in(0,\infty)$, if $t<t'$, then $t\errsymb{s}=t'\errsymb{s}$.
    \item[(E8)] $\abs{1+\errsymb{s}}=1+\errsymb{s}$
\end{enumerate}
\end{proposition}
\begin{proof}
This follows from~\cite[Lemmas~3.1 and~3.3]{Higham96}
\end{proof}

The definition and properties of $\errsymb{~}$  
follow the lines of classical error analysis, as e.g., in~\cite[Chapter~3]{Higham96}. Our presentation may differ in 
minor details which we have chosen for our own convenience. 
\rojito{In all what follows, the round-off unit $\bfu$ is always 
sufficiently small, so that the inequalities $t\bfu<1/2$ hold 
true for the values of $t$ at hand. As is customary in
finite-precision analyses, we won't explicitly point 
to these bounds.}

\subsection{Basic finite precision algorithms}

The following two propositions show the nice properties of the numerical computations that underlie the algorithm~\nameref{alg:PVAlgorithmFP}. Their statements refer 
to three aspects: 1) the number of arithmetic operations performed, 2) error estimates for a given input, 
and 3) error estimates for approximate inputs. From these  bounds we can obtain bit-complexity estimates, as 
floating-point operations take $\Oh(\abs{\log\bfu}^2)$-time 
(this being non-tight, one can obtain better bounds 
using fast multiplication algorithms). 

\rojito{An algorithm computing inner products with sharper 
error bounds was recently analyzed in~\cite{blanchard2020} (see also \cite{jeannerod2016} for a survey on another family of recent improvements in this respect). For 
our purposes, however, the simpler Proposition~\ref{prop:error-innerproduct} is sufficient.}

\begin{proposition}\label{prop:error-innerproduct}
There is a numerical algorithm which, with input $x,y\in\bbR^m$,
computes $\langle x,y\rangle$. This algorithm satisfies the following:
\begin{enumerate}[(i)]
    \item It performs $\Oh(m)$ arithmetic operations.
    \item On input $x,y\in\bbF^m$, the computed value
$\fl(\langle x,y\rangle)$ satisfies
\begin{equation}
    \fl(\langle x,y\rangle)=\langle x,y\rangle+\langle \abs{x},\abs{y}\rangle\errsymb{\log m+2},
\end{equation}
where $\abs{x}=(\abs{x_1},\ldots,\abs{x_n})$.
\item Assume $\tilde{x},\tilde{y}\in\bbF^m$ and 
$x,y\in\bbR^m$ are such that, for all $i$,
\[\tilde{x}_i=x_i+t_i\errsymb{\epsilon}\text{ and }\tilde{y}_i=y_i+t_i'\errsymb{\epsilon'}\]
for some $t,t'\in[0,\infty)^m$ and $\epsilon,\epsilon'\geq 1$. Then the computed value
$\fl(\langle \tilde{x},\tilde{y}\rangle)$ satisfies
\begin{multline*}
    \fl(\langle \tilde{x},\tilde{y}\rangle)=\langle x,y\rangle\\+\max\{\langle \abs{x},\abs{y}\rangle,\langle \abs{t},\abs{y}\rangle,\langle \abs{x},\abs{t'}\rangle,\langle \abs{t},\abs{t'}\rangle\}\errsymb{\log m+\epsilon+\epsilon'+2}.
\end{multline*}
\end{enumerate}
\end{proposition}

\begin{proposition}\label{prop:error-norm}
There is a numerical algorithm which, with input $x\in\bbR^m$,
computes $\|x\|$. This algorithm satisfies the following:
\begin{enumerate}[(i)]
    \item It performs $\Oh(m)$ arithmetic operations.
    \item On input $x\in\bbF^m$, the computed value
$\fl(\|x\|)$ satisfies
\begin{equation}
    \fl(\|x\|)=\|x\|(1+\errsymb{\log m+3}).
\end{equation}
\item Assume $\tilde{x}\in\bbF^m$ and $x\in\bbR^m$ are 
such that, for all $i$,
\[\tilde{x}_i=x_i+t_i\errsymb{\epsilon}\]
for some $t\in[0,\infty)^m$ and $\epsilon\geq 1$. Then the computed value
$\fl(\|\tilde{x}\|\rangle)$ satisfies
\[
\fl(\|\tilde{x}\|)=\|x\|+\max\{\|x\|,\|t\|\}\errsymb{\log m+\epsilon+3}.
\]
\end{enumerate}
\end{proposition}

\begin{proposition}\label{prop:error-xalpha}
There is a numerical algorithm which, with input $x\in\bbR^n$ and $\alpha\in\bbN$,
computes $x^{\alpha}$. This algorithm satisfies the following:
\begin{enumerate}[(i)]
    \item It performs $\Oh(\log \abs{\alpha})$ arithmetic operations.
    \item On input $x\in\bbF^n$, the computed value
$\fl(x^\alpha)$ satisfies
\begin{equation*}
   \fl(x^\alpha)=\begin{cases}x^\alpha(1+\errsymb{\abs{\alpha}-1})&\text{if }\abs{\alpha}>1\\x^\alpha,&\text{otherwise}.
   \end{cases}
\end{equation*}
\item Assume that $\tilde{x}\in\bbF^n$ and $x\in\bbR^n$ 
are such that, for all $i$,
\[\tilde{x}_i=x_i(1+\errsymb{\epsilon})\]
for some $t\in[0,\infty)^m$ and $\epsilon\geq 1$. Then the computed value
$\fl(\tilde{x}^\alpha)$ satisfies
\[
\fl(\tilde{x}^\alpha)=\begin{cases}x^{\alpha}(1+\errsymb{\abs{\alpha}(1+\epsilon)-1}),&\text{if }\alpha\neq 0\\1,&\text{otherwise}.\end{cases}
\]
\end{enumerate}
\end{proposition}

\begin{proof}[Proof of Proposition~\ref{prop:error-innerproduct}]
The algorithm will first perform all the products $x_iy_i$ and them perform their sum by recursively dividing the sum into 
\[
  \sum_{i\in I}x_iy_i+\sum_{i\in I^\complement}x_iy_i
\]
where $I$ and its complement, $I^\complement$ have size almost equal, differing in at most one.

(i) We initially perform $m$ products and then $m-1$ additions. Note that the latter is independent of how we achieve the 
final sum, we sum as we do to minimize the error.

(ii) We will prove using induction the stronger claim that for the above algorithm
\[
 \fl(\langle x,y\rangle)=\langle x,y\rangle+\langle \abs{x},\abs{y}\rangle\errsymb{\lceil\log m\rceil+1}
\]
where $\lceil x\rceil$ is the minimum integer bigger or equal than $x$. Note that the claim is true for $m=1$ and $m=2$.

By the recursive nature of the algorithm, we have that
\begin{align*}
    &\fl\left(\sum_{i=1}^mx_iy_i\right)&\\
    &=\fl\left(\sum_{i\in I}x_iy_i\right)\,\widetilde{+}\,\fl\left(\sum_{i\in I^\complement}x_iy_i\right)\\
    &=\left(\sum_{i\in I}x_iy_i+\left(\sum_{i\in I}\abs{x_i}\abs{y_i}\right)\errsymb{\lceil \log\abs{I}\rceil+1}\right.\\&\left.~+\sum_{i\in I^\complement}x_iy_i+\left(\sum_{i\in I^\complement}\abs{x_i}\abs{y_i}\right)\errsymb{\lceil \log(n-\abs{I})\rceil+1}\right)(1+\errsymb{1})&\text{(Induction)}\\
    &=\left(\sum_{i=1}^nx_iy_i+\left(\sum_{i=1}^n\abs{x_i}\abs{y_i}\right) \errsymb{\log\max\{\abs{I},n-\abs{I}\}+1}   \right)(1+\errsymb{1})&(E6)\\
    &=\left(\langle x,y\rangle+\langle\abs{x},\abs{y}\rangle\errsymb{\lceil\log\max\{\abs{I},n-\abs{I}\}\rceil+1}\right)(1+\errsymb{1})
\end{align*}
Now, when $\abs{I}$ and $n-\abs{I}$ differ in at most one, we have that
\[
\lceil\log\max\{\abs{I},n-\abs{I}\}\rceil+1\leq \lceil\log n\rceil.
\]
Thus
\begin{align*}
    &=\left(\langle x,y\rangle+\langle\abs{x},\abs{y}\rangle\errsymb{\lceil\log n\rceil}\right)(1+\errsymb{1})\\
    &=\langle x,y\rangle+\langle x,y\rangle\errsymb{1}+\langle\abs{x},\abs{y}\rangle(\errsymb{\lceil\log n\rceil}+\errsymb{\lceil\log n\rceil}\errsymb{1})\\
    &=\langle x,y\rangle+\langle \abs{x},\abs{y}\rangle\errsymb{1}+\langle\abs{x},\abs{y}\rangle(\errsymb{\lceil\log n\rceil}+\errsymb{\lceil\log n\rceil}\errsymb{1})&\langle x,y\rangle \leq \langle \abs{x},\abs{y}\rangle\\
    &=\langle x,y\rangle+\langle\abs{x},\abs{y}\rangle(\errsymb{\lceil\log n\rceil}+\errsymb{1}+\errsymb{\lceil\log n\rceil}\errsymb{1})&(E1)\\
    &=\langle x,y\rangle+\langle\abs{x},\abs{y}\rangle\errsymb{\lceil\log n\rceil+1}&(E2).
\end{align*}

(iii) Note that
\begin{align*}
    \langle \tilde{x},\tilde{y}\rangle&=\langle x,y\rangle +\langle (t_i\errsymb{\epsilon}),y\rangle+\langle x,(t_i'\errsymb{\epsilon'})\rangle+\langle (t_i\errsymb{\epsilon}),(t_i'\errsymb{\epsilon'})\rangle\\
    &=\langle x,y\rangle+\langle \abs{t},\abs{y}\rangle \errsymb{\epsilon}+\langle \abs{x},\abs{t'}\rangle \errsymb{\epsilon'}+\langle \abs{t},\abs{t'}\rangle \errsymb{\epsilon}\errsymb{\epsilon'}&(E6)\\
    &=\langle x,y\rangle+\max\{\langle \abs{t},\abs{y}\rangle,\langle \abs{x},\abs{t'}\rangle,\langle \abs{t},\abs{t'}\rangle\}(\errsymb{\epsilon}+\errsymb{\epsilon'}+\errsymb{\epsilon}\errsymb{\epsilon'})&(E7)\\
    &=\langle x,y\rangle+\max\{\langle \abs{t},\abs{y}\rangle,\langle \abs{x},\abs{t'}\rangle,\langle \abs{t},\abs{t'}\rangle\}\errsymb{\epsilon+\epsilon'}&(E2)
\end{align*}
An analogous statement holds for $\langle \abs{\tilde{x}},\abs{\tilde{y}}\rangle$. Now, combining this and (ii), we get that{\small
\begin{align*}
    &\fl(\langle \tilde{x},\tilde{y}\rangle)\\&=\langle \tilde{x},\tilde{y}\rangle+\langle \abs{\tilde{x}},\abs{\tilde{y}}\rangle\errsymb{\log m+2}\\
    &=\langle x,y\rangle+\max\{\langle \abs{t},\abs{y}\rangle,\langle \abs{x},\abs{t'}\rangle,\langle \abs{t},\abs{t'}\rangle\}\errsymb{\epsilon+\epsilon'}\\&~~+\left(\langle \abs{x},\abs{y}\rangle+\max\{\langle \abs{t},\abs{y}\rangle,\langle \abs{x},\abs{t'}\rangle,\langle \abs{t},\abs{t'}\rangle\}\errsymb{\epsilon+\epsilon'}\right)\errsymb{\log m+2}\\
    &=\langle x,y\rangle+\max\{\langle \abs{x},\abs{y}\rangle,\langle \abs{t},\abs{y}\rangle,\langle \abs{x},\abs{t'}\rangle,\langle \abs{t},\abs{t'}\rangle\}\\&~~~~~~~~~~~~~~~~\cdot\left(\errsymb{\epsilon+\epsilon'}+\errsymb{\log m+2}+\errsymb{\epsilon+\epsilon'}\errsymb{\log m+2}\right)&(E7)\\
    &=\langle x,y\rangle+\max\{\langle \abs{x},\abs{y}\rangle,\langle \abs{t},\abs{y}\rangle,\langle \abs{x},\abs{t'}\rangle,\langle \abs{t},\abs{t'}\rangle\}\errsymb{\log m+\epsilon+\epsilon'+2}.
\end{align*}}
\end{proof}

\begin{proof}[Proof of Proposition~\ref{prop:error-norm}]
The proof is analogous to that of Proposition~\ref{prop:error-innerproduct}.
\end{proof}

\begin{proof}[Proof of Proposition~\ref{prop:error-xalpha}]
The proof is analogous to that of Proposition~\ref{prop:error-innerproduct}, but we have to take into account that errors accumulate additively since in 
each multiplication the errors of the computed quantities are added by (E2).
\end{proof}

\subsection{The final proofs}
The following lemma is useful.

\begin{lemma}\label{lem:error-weylnorm}
There is a numerical algorithm which, with input $f\in\Pd$, computes the Weyl norm $\|f\|$ of $f$. This algorithm 
performs $\Oh(N)$ arithmetic operations, and, on input $f\in\Pd\cap\bbF[X_1,\ldots,X_n]$, the computed value $\fl(\|f\|)$ satisfies
\[\fl(\|f\|)=\|f\|(1+\errsymb{\log N+8}).\]
Moreover, for general $f\in\Pd$,
\[\fl(\|r(f)\|)=\|f\|(1+\errsymb{\log N+9}).\]
\end{lemma}

\begin{proof}
To compute the Weyl norm, we first compute the vector $\left(\binom{d}{\alpha}^{-1/2}f_\alpha\right)$ and then its norm. To compute the vector, we take the floating point approximation of $\binom{d}{\alpha}$, we compute its square root and we divide $f_\alpha$ by the computed square root. Hence
\begin{align*}
    \fl\left(\binom{d}{\alpha}^{-1/2}f_\alpha\right)&=\binom{d}{\alpha}^{-1/2}f_\alpha\frac{(1+\errsymb{1})}{\sqrt{1+\errsymb{1}}(1+\errsymb{1}}\\
    &=\binom{d}{\alpha}^{-1/2}f_\alpha(1+\errsymb{5})&\text{(Proposition~\ref{prop:Higham})}
\end{align*}
Now, the lemma follows from Proposition~\ref{prop:error-norm}.
\end{proof}

We can now give the proofs of Propositions \ref{prop:error-absolutevalueofevaluation} and \ref{prop:error-normgradientvector}.

\begin{proof}[Proof of Proposition~\ref{prop:error-absolutevalueofevaluation}]
We first compute $f(x)$ as $\langle (f_\alpha),(x^\alpha))\rangle$, where the $x^{\alpha}$ are 
computed one by one, and then divide the result by the computed $\|f\|\|(1,x)\|^{d-1}$ to obtain $\hat{f}(x)$.

By Propositions~\ref{prop:error-innerproduct} and~\ref{prop:error-xalpha} and (E7), we have that
\[
 \fl(f(x))=f(x)+\|f\|\|(1,x)\|^d\errsymb{\log N+d+1},
\]
since $\langle (\abs{f_{\alpha}}),(\abs{x^{\alpha}})\rangle =g(\abs{x})$, where $g=\sum_\alpha \abs{f_{\alpha}}X^\alpha$, is bounded by $\|f\|\|(1,x)\|^d$, by Lemma~\ref{lem:lipschitz}.

Also, by Proposition~\ref{prop:error-norm}, Lemma~\ref{lem:error-weylnorm} and (E2), we have that
\[\fl{\|f\|\|(1,x)\|^{d-1}}=\|f\|\|(1,x)\|^{d-1}(1+\errsymb{\log N+d\log (n+1)+4d+2}).\]

Now, $N\leq (n+1)^d$. Thus we have that
\[
\fl(f(x))=f(x)+\|f\|\|(1,x)\|^d\errsymb{3d\log(n+1)}
\]
and
\[\fl(\|f\|\|(1,x)\|^{d-1})=\|f\|\|(1,x)\|^{d-1}(1+\errsymb{8d\log (n+1)}).\]
Although, doing this we are not obtaining tight bounds, we have to recall that the number of digits is proportional to the logarithm of what is inside $\errsymb{\cdot}$.

To finish, we only have to do the division. Thus{\footnotesize
\begin{align*}
    &\fl(\hat{f}(x))\\
    &=\fl(f(x))/\fl(\|f\|\|(1,x)\|^{d-1})(1+\errsymb{1})\\
    &=(f(x)+\|f\|\|(1,x)\|^d\errsymb{3d\log(n+1)})/\left(\|f\|\|(1,x)\|^{d-1}(1+\errsymb{7d\log (n+1)})\right)(1+\errsymb{1})\\
    &=(\hat{f}(x)+\|(1,x)\|\errsymb{3d\log(n+1)}(1+\errsymb{8d\log (n+1)})^{-1}(1+\errsymb{1})\\
    &=(\hat{f}(x)+\|(1,x)\|\errsymb{3d\log(n+1)}(1+\errsymb{16d\log (n+1)+1})\\
    &=\hat{f}(x)+\hat{f}(x)\errsymb{10d\log (n+1)+1}\\&~~+\|(1,x)\|(\errsymb{3d\log(n+1)}+\errsymb{3d\log(n+1)}\errsymb{14d\log (n+1)+1})\\
    &=\hat{f}(x)\\&~~+\|(1,x)\|(\errsymb{16d\log (n+1)+1}+\errsymb{3d\log(n+1)}+\errsymb{3d\log(n+1)}\errsymb{16d\log (n+1)+1})\\
    &=\hat{f}(x)+\|(1,x)\|\errsymb{19d\log (n+1)+1}\\
    &=\hat{f}(x)+\|(1,x)\|\errsymb{20d\log (n+1)}
\end{align*}}
where the first equality follows from the way we compute $\hat{f}(x)$, 
the second one from the above identities, the fourth one from (E3) 
and (E2), the sixth one from Lemma~\ref{lem:lipschitz} and (E7), 
the eighth one from (E2), and the last one from (E1).

The result for $r(f)$ and $r(x)$ follows similarly.
\end{proof}

\begin{proof}[Proof of Proposition~\ref{prop:error-normgradientvector}]
We compute each $\partial_jf(x)$ as we computed $f(x)$. After that, we compute $\|\partial f(x)\|$, $d\|f\|\|(1,x)\|^{d-2}$ and their quotient.

By Propositions~\ref{prop:error-innerproduct} and~\ref{prop:error-xalpha} and (E7), we have that
\[
\fl(\partial_jf(x))=\partial_jf(x)+\partial_jg(\abs{x})\errsymb{\log N+d+1},
\]
where $g=\sum_\alpha \abs{f_{\alpha}}X^{\alpha}$. Now, by Proposition~\ref{prop:error-norm}, we have that
\[
\fl(\|\partial f(x)\|)=\|\partial f(x)\|+\max\{\|\partial f(x)\|,\partial g(\abs{x})\|\}\errsymb{\log N+\log n+d+4}.
\]
However, by Lemma~\ref{lem:lipschitz}, both 
$\|\partial f(x)\|$ and $\|\partial g(\abs{x})\|$ are bounded by $d\|f\|\|(1,x)\|^{d-1}$. Thus, by (E7), 
\[
\fl(\|\partial f(x)\|)=\|\partial f(x)\|+d\|f\|\|(1,x)\|^{d-1}\errsymb{\log N+\log n+d+4}.
\]

Again, by Proposition~\ref{prop:error-norm}, Lemma~\ref{lem:error-weylnorm} and (E2), we have that
\[\fl\left(d\|f\|\|(1,x)\|^{d-2}\right)
=d\|f\|\|(1,x)\|^{d-2}(1+\errsymb{\log N+d\log (n+1)+4d+2}).\]

Now, as $N\leq (n+1)^d$, we have
\[
\fl(\|\partial f(x)\|)=\|\partial f(x)\|+d\|f\|\|(1,x)\|^{d-1}\errsymb{7d\log(n+1)}
\]
and
\[\fl\left(d\|f\|\|(1,x)\|^{d-2}\right)
=d\|f\|\|(1,x)\|^{d-2}(1+\errsymb{8d\log(n+1)}).\]

Now, arguing as in Proposition~\ref{prop:error-absolutevalueofevaluation}, the desired statement follows.
\end{proof}

\end{appendices}

{\small
\bibliographystyle{plain}
\bibliography{biblio.bib}
}
\end{document}